\documentclass[11pt]{article}

\usepackage{amsmath}
\usepackage{amssymb}
\usepackage{amsfonts}
\usepackage{graphicx}
\usepackage{color}
\usepackage{amsthm} 
\usepackage[colorlinks=true]{hyperref}

\def\hat{\widehat}
\def\phi{\varphi}

\newtheorem{theorem}{Theorem}[section]
\newtheorem{lemma}[theorem]{Lemma}

\newtheorem{corollary}[theorem]{Corollary}
\newtheorem{claim}[theorem]{Claim}
\newtheorem{observation}[theorem]{Observation}

\newcounter{rnc}

\def\B{{\cal B}}

\def\M{\mathbf{M}}
\def\K{\mathbf{K}}
\def\R{\mathbb{R}}
\def\rank{\mathrm{rank}\,}
\def\Pf{\mathrm{Pf}\,}
\def\deg{\mathrm{deg}}
\def\odd{\mathrm{odd}}
\def\Search{{\sf Search}}
\def\Blossom{{\sf Blossom}}
\def\DBlossom{{\sf Graft}}
\def\Expand{{\sf Expand}}
\def\Augment{{\sf Augment}}

\newcommand{\X}[1]{}

\title{A Weighted Linear Matroid Parity Algorithm\thanks{
A preliminary version of this paper has
appeared in Proceedings of the 49th
Annual ACM Symposium on Theory of Computing
(STOC 2017), pp.~264--276.}}
\author{Satoru Iwata
\thanks{Department of Mathematical Informatics, University of Tokyo, Tokyo 113-8656, Japan.
E-mail: iwata@mist.i.u-tokyo.ac.jp} 
\and Yusuke Kobayashi
\thanks{Research Institute for Mathematical Sciences, Kyoto University, Kyoto, 606-8502, Japan.
E-mail: yusuke@kurims.kyoto-u.ac.jp}
}
\begin{document}
\maketitle
\begin{abstract}
The matroid parity (or matroid matching) problem, introduced as a common generalization of matching 
and matroid intersection problems, is so general that it requires an exponential number of oracle calls. 
Nevertheless, Lov\'asz (1980) showed that this problem admits a min-max formula and a polynomial algorithm for 
linearly represented matroids. Since then efficient algorithms have been developed for 
the linear matroid parity problem. 

In this paper, we present a combinatorial, deterministic, polynomial-time algorithm for 
the weighted linear matroid parity problem. The algorithm builds on a polynomial matrix 
formulation using Pfaffian and adopts a primal-dual approach based on 
the augmenting path algorithm of Gabow and Stallmann (1986) for the unweighted problem.  
\end{abstract}

\newpage

\section{Introduction}
The matroid parity problem \cite{Law76} (also known as the matchoid problem~\cite{Jen74} or the matroid matching problem~\cite{Lov78}) 
was introduced as a common generalization of matching and matroid intersection problems. 
In the general case, it requires an exponential number of independence oracle calls \cite{JK82,Lov80a}, 
and a PTAS has been developed only recently \cite{LSV13}.  
Nevertheless, Lov\'asz~\cite{Lov78,Lov80a,Lov80b} showed that the problem admits 
a min-max theorem for linear matroids and presented a polynomial algorithm that 
is applicable if the matroid in question is represented by a matrix. 

Since then, efficient combinatorial algorithms have been developed for 
this linear matroid parity problem \cite{GS86,Orl08,OV92}. 
Gabow and Stallmann~\cite{GS86} developed an augmenting path algorithm with 
the aid of a linear algebraic trick, 
which was later extended to the linear delta-matroid parity problem~\cite{GIM03}. 
Orlin and Vande Vate~\cite{OV92} provided 
an algorithm that solves this problem by repeatedly solving matroid intersection 
problems coming from the min-max theorem. Later, Orlin~\cite{Orl08} improved the 
running time bound of this algorithm. The current best deterministic running time 
bound due to \cite{GS86,Orl08} is $O(nm^\omega)$, where $n$ is the cardinality of 
the ground set, $m$ is the rank of the linear matroid, and $\omega$ is 
the matrix multiplication exponent, which is at most $2.38$.  
These combinatorial algorithms, however, tend to be complicated. 

An alternative approach that leads to simpler randomized algorithms is based on 
an algebraic method. This is originated by Lov\'asz~\cite{Lov79}, who formulated 
the linear matroid parity problem as rank computation of a skew-symmetric matrix 
that contains independent parameters. Substituting randomly generated numbers to 
these parameters enables us to compute the optimal value with high probability. 
A straightforward adaptation of this approach requires iterations to find an optimal 
solution. Cheung, Lau, and Leung \cite{CLL14} have improved this algorithm to 
run in $O(nm^{\omega-1})$ time, extending the techniques of 
Harvey~\cite{Har09} developed for matching and matroid intersection. 

While matching and matroid intersection algorithms~\cite{Edm65,Edm68} have been successfully extended 
to their weighted version~\cite{Edm65b,Edm79,IT76,Law75}, 
no polynomial algorithms have been known for the 
weighted linear matroid parity problem for more than three decades.
Camerini, Galbiati, and Maffioli~\cite{CGM92} developed a random pseudopolynomial 
algorithm for the weighted linear matroid parity problem by introducing 
a polynomial matrix formulation that extends the matrix formulation of 
Lov\'asz~\cite{Lov79}. This algorithm was later improved by 
Cheung, Lau, and Leung~\cite{CLL14}. The resulting complexity, however,  
remained pseudopolynomial. 
Tong, Lawler, and Vazirani~\cite{TLV84} observed that 
the weighted matroid parity problem on gammoids can be solved in polynomial time 
by reduction to the weighted matching problem. 
As a relaxation of the matroid matching polytope, 
Vande Vate~\cite{Van92} introduced the fractional matroid matching polytope. 
Gijswijt and Pap~\cite{GP13} devised a polynomial algorithm for optimizing 
linear functions over this polytope. The polytope was shown to be half-integral, 
and the algorithm does not necessarily yield an integral solution. 
 
This paper presents a combinatorial, deterministic, polynomial-time algorithm 
for the weighted linear matroid parity problem. To do so, we combine 
algebraic approach and augmenting path technique
together with the use of node potentials. 
The algorithm builds on a polynomial matrix formulation, which naturally extends 
the one discussed in \cite{GI05} for the unweighted problem. The algorithm employs 
a modification of the augmenting path search procedure for the unweighted 
problem by Gabow and Stallmann~\cite{GS86}. It adopts a primal-dual approach 
without writing an explicit LP description. The correctness proof 
for the optimality is based on the idea of combinatorial relaxation for 
polynomial matrices due to Murota~\cite{Mur95}.
The algorithm is shown to require $O(n^3m)$ arithmetic operations. 
This leads to a strongly polynomial algorithm for linear matroids represented over a finite field.  
For linear matroids represented over the rational field, 
one can exploit our algorithm to solve the problem in polynomial time. 

Independently of the present work, Gyula Pap has obtained another 
combinatorial, deterministic, polynomial-time algorithm for the 
weighted linear matroid parity problem based on a different approach.  

The matroid matching theory of Lov\'asz \cite{Lov80b} in fact deals with 
a more general class of matroids that enjoy the double circuit property. 
Dress and Lov\'asz \cite{DL87} showed that algebraic matroids satisfy this property. 
Subsequently, Hochst\"attler and Kern \cite{HK89} showed the same phenomenon for 
pseudomodular matroids. The min-max theorem follows for this class of matroids. 
To design a polynomial algorithm, however, one has to establish how to represent 
those matroids in a compact manner. Extending this approach to the weighted problem 
is left for possible future investigation. 

The linear matroid parity problem finds various applications: 
structural solvability analysis of passive electric networks \cite{Mil74}, 
pinning down planar skeleton structures \cite{LP86}, and 
maximum genus cellular embedding of graphs \cite{FGM88}. 
We describe below two interesting applications of the weighted matroid parity problem
in combinatorial optimization. 

A $T$-path in a graph is a path between two distinct vertices in the terminal set $T$. 
Mader~\cite{Mad78} showed a min-max characterization of the maximum number 
of openly disjoint $T$-paths. The problem can be equivalently formulated 
in terms of ${\mathcal S}$-paths, where ${\mathcal S}$ is a partition of $T$ 
and an ${\mathcal S}$-path is a $T$-path between two different components of 
${\mathcal S}$. Lov\'asz~\cite{Lov80b} formulated the problem as a matroid 
matching problem and showed that one can find a maximum number of disjoint 
${\mathcal S}$-paths in polynomial time. Schrijver~\cite{Sch03} has described
a more direct reduction to the linear matroid parity problem. 

The disjoint ${\cal S}$-paths problem has been extended to path packing problems 
in group-labeled graphs \cite{CCG08,CGGGLS06,Pap07}. 
Tanigawa and Yamaguchi~\cite{TY16} have shown that these problems also reduce 
to the matroid matching problem with double circuit property. Yamaguchi~\cite{Yam16a} 
clarifies a characterization of the groups for which those problems reduce 
to the linear matroid parity problem. 

As a weighted version of the disjoint ${\mathcal S}$-paths problem, it is quite 
natural to think of finding disjoint ${\mathcal S}$-paths of minimum total length. 
It is not immediately clear that this problem reduces to the weighted linear 
matroid parity problem. A recent paper of Yamaguchi~\cite{Yam16b} clarifies that 
this is indeed the case. He also shows that the reduction results on 
the path packing problems on group-labeled graphs also extend to the 
weighted version. 

The weighted linear matroid parity is also useful in the design of 
approximation algorithms. Pr\"omel and Steger~\cite{PS00} provided  
an approximation algorithm for the Steiner tree problem. 
Given an instance of the Steiner tree problem, construct a hypergraph 
on the terminal set such that each hyperedge corresponds to a terminal subset 
of cardinality at most three and regard the shortest length of a Steiner tree 
for the terminal subset as the cost of the hyperedge. The problem of finding 
a minimum cost spanning hypertree in the resulting hypergraph can be 
converted to the problem of finding a minimum spanning tree in a 3-uniform hypergraph, 
which is a special case of the weighted parity problem for graphic matroids. 
The minimum spanning hypertree thus obtained costs at most 5/3 of the optimal 
value of the original Steiner tree problem, and one can construct 
a Steiner tree from the spanning hypertree without increasing the cost. 
Thus they gave a 5/3-approximation algorithm for the Steiner tree problem 
via weighted linear matroid parity. This is a very interesting approach that 
suggests further use of weighted linear matroid parity in the design of 
approximation algorithms, even though the performance ratio is larger than 
the current best one for the Steiner tree problem \cite{BGRS13}.

\section{The Minimum-Weight Parity Base Problem}
\label{sec:problemdef}

Let $A$ be a matrix of row-full rank over an arbitrary field $\K$
with row set $U$ and column set $V$. Assume that both $m=|U|$ and $n=|V|$ are even. 
The column set $V$ is partitioned into pairs, called {\em lines}. 
Each $v\in V$ has its {\em mate} $\bar{v}$ such that $\{v,\bar{v}\}$ 
is a line. We denote by $L$ the set of lines, and suppose that 
each line $\ell\in L$ has a weight $w_\ell \in \R$.

The linear dependence of the column vectors naturally defines a matroid $\M(A)$ on $V$. Let   
$\B$ denote its base family. A base $B\in\B$ is called a {\em parity base} if it consists of lines. 
As a weighted version of the linear matroid parity problem, we will consider the problem of 
finding a parity base of minimum weight, where the weight of a parity base is the sum of the 
weights of lines in it. We denote the optimal value by $\zeta(A,L,w)$. This problem generalizes 
finding a minimum-weight perfect matching in graphs and a minimum-weight common base of a pair 
of linear matroids on the same ground set. 

As another weighted version of the matroid parity problem, one can think of finding 
a matching (independent parity set) of maximum weight. This problem can be easily 
reduced to the minimum-weight parity base problem. 

Associated with the minimum-weight parity base problem, we consider a skew-symmetric polynomial matrix 
$\Phi_A(\theta)$ in variable $\theta$ defined by 
$$\Phi_A(\theta)=\begin{pmatrix} O & A \\ -A^\top & D(\theta) \end{pmatrix},$$
where $D(\theta)$ is a block-diagonal matrix in which each block is a $2\times 2$ skew-symmetric polynomial matrix 
$D_\ell(\theta)=\begin{pmatrix} 0 & -\tau_\ell\theta^{w_\ell} \\ \tau_\ell\theta^{w_\ell} & 0 \end{pmatrix}$ 
corresponding to a line $\ell\in L$. Assume that the coefficients $\tau_\ell$ are independent parameters 
(or indeterminates). 

For a skew-symmetric matrix $\Phi$ whose rows and columns are indexed by $W$, 
the {\em support graph} of $\Phi$ is the graph $\Gamma=(W,E)$ with edge set 
$E= \{(u, v) \mid \Phi_{uv}\neq 0\}$. 
We denote by $\Pf\Phi$ the {\em Pfaffian} of $\Phi$, which is defined as follows: 
$$
\Pf\Phi = \sum_{M} \sigma_M \prod_{(u, v) \in M} \Phi_{uv},
$$
where the sum is taken over all perfect matchings $M$ in $\Gamma$ and
$\sigma_M$ takes $\pm 1$ in a suitable manner, see~\cite{LP86}.
It is well-known that $\det \Phi=(\Pf \Phi)^2$ and 
$\Pf(S\Phi S^\top)=\Pf\Phi\cdot\det S$ for any square matrix $S$. 

We have the following lemma that associates the optimal value of the minimum-weight
parity base problem with $\Pf\Phi_A(\theta)$. 

\begin{lemma}
\label{lem:Pf}
The optimal value of the minimum-weight parity base problem is given by 
$$\zeta(A,L,w)=\sum_{\ell\in L}w_\ell-\deg_\theta\,\Pf\Phi_A(\theta).$$ 
In particular, if $\Pf\Phi_A(\theta)=0$ (i.e., $\deg_\theta\,\Pf\Phi_A(\theta) = -\infty$), then there is no parity base. 
\end{lemma}

\begin{proof}
We split $\Phi_A(\theta)$ into $\Psi_A$ and $\Delta(\theta)$
such that 
\begin{align*}
& \Phi_A(\theta) = \Psi_A + \Delta(\theta), &
&\Psi_A=\begin{pmatrix} O & A \\ -A^\top & O \end{pmatrix}, &
\Delta(\theta)=\begin{pmatrix} O & O \\ O & D(\theta) \end{pmatrix}.
\end{align*}
The row and column sets of these skew-symmetric matrices are indexed by $W:=U\cup V$. 
By \cite[Lemma 7.3.20]{Mur00}, we have 
$$\Pf \Phi_A(\theta) = \sum_{X \subseteq W} \pm \Pf \Psi_A[W \setminus X] 
\cdot \Pf \Delta(\theta)[X],$$
where each sign is determined by the choice of $X$,  
$\Delta(\theta)[X]$ is the principal submatrix of $\Delta(\theta)$ 
whose rows and columns are both indexed by $X$, 
and $\Psi_A[W \setminus X]$ is defined in a similar way. 
One can see that $\Pf \Delta(\theta)[X] \not= 0$ if and only if $X \subseteq V$ 
(or, equivalently $B:=V\setminus X$) is a union of lines.
One can also see for $X \subseteq V$ that $\Pf \Psi_A[W \setminus X] \not= 0$ if and only if
$A[U, V \setminus X]$ is 
nonsingular, which means that 
$B$ is a base of $\M(A)$. 
Thus, we have 
$$
\Pf \Phi_A (\theta) = \sum_{B} \pm \Pf \Psi_A[U \cup B] \cdot \Pf \Delta(\theta)[V \setminus B], 
$$
where the sum is taken over all parity bases $B$. 
Note that no term is canceled out in the summation, 
because each term contains a distinct set of independent parameters.
For a parity base $B$, we have 
$$\deg_\theta (\Pf \Psi_A[U \cup B] \cdot \Pf \Delta(\theta)[V \setminus B]) 
= \sum_{\ell\subseteq V \setminus B}w_\ell = \sum_{\ell\in L}w_\ell - \sum_{\ell\subseteq B}w_\ell,$$ 
which implies that 
the minimum weight of a parity base is  
$\displaystyle\sum_{\ell\in L}w_\ell-\deg_\theta\,\Pf \Phi_A (\theta)$. 
\end{proof}

Note that Lemma~\ref{lem:Pf} does not immediately lead to a (randomized) polynomial-time 
algorithm for the minimum weight parity base problem. 
This is because computing the degree of the Pfaffian of a skew-symmetric 
polynomial matrix is not so easy. 
Indeed, the algorithms in \cite{CGM92,CLL14} for the weighted linear matroid parity problem
compute the degree of the Pfaffian of another skew-symmetric polynomial matrix, 
which results in pseudopolynomial complexity.

\section{Algorithm Outline}
\label{sec:algorithm}
In this section, we describe the outline of our algorithm for solving the minimum-weight 
parity base problem. 

We regard the column set $V$ as a vertex set. 
The algorithm works on a vertex set $V^*\supseteq V$ that includes some new vertices 
generated during the execution. 
The algorithm keeps a nested (laminar) collection 
$\Lambda=\{H_1,\ldots,H_{|\Lambda|}\}$ of vertex subsets of $V^*$ 
such that $H_i \cap V$ is a union of lines for each $i$.
The indices satisfy that, for any two members $H_i, H_j \in \Lambda$ with $i < j$, 
either $H_i \cap H_j = \emptyset$ or $H_i \subsetneq H_j$ holds. 
Each member of $\Lambda$ is called a {\em blossom}.
The algorithm maintains a potential $p:V^*\to\R$ 
and a nonnegative variable $q:\Lambda\to\R_+$, which 
are collectively called {\em dual variables}. 

We note that although $p$ and $q$ are called dual variables, they do not correspond to 
dual variables of an LP-relaxation of the minimum-weight parity base problem. 
Indeed, this paper presents neither an LP-formulation nor a min-max formula for 
the minimum-weight parity base problem, explicitly.
We will show instead that one can obtain a parity base $B$ that admits feasible dual variables $p$ and $q$, 
which provide a certificate for the optimality of $B$. 

The algorithm starts with splitting the weight $w_\ell$ into $p(v)$ and $p(\bar{v})$ for each 
line $\ell=\{v,\bar{v}\}\in L$, i.e., $p(v) + p(\bar{v}) = w_\ell$. 
Then it executes the greedy algorithm for finding a base $B\in\B$ 
with minimum value of $p(B)=\sum_{u\in B}p(u)$. If $B$ is a parity base, then $B$ is obviously 
a minimum-weight parity base. 
Otherwise, there exists
a line $\ell=\{v,\bar{v}\}$ in which exactly one of its two vertices belongs to $B$. 
Such a line is called a {\em source line} and each vertex in a source line is called 
a {\em source vertex}. A line that is not a source line is called a {\em normal line}. 

The algorithm initializes $\Lambda:=\emptyset$ and 
proceeds iterations of primal and dual updates, 
keeping dual feasibility. 
In each iteration, the algorithm applies the breadth-first search to find an augmenting path. 
In the meantime, the algorithm sometimes detects a new blossom and adds it to $\Lambda$.
If an augmenting path $P$ is found, the algorithm updates $B$ along $P$. 
This will reduce the number of source lines by two. 
If the search procedure terminates without finding an augmenting path, 
the algorithm updates the dual variables to create new tight edges.  
The algorithm repeats this process until $B$ becomes a parity base.  
Then $B$ is a minimum-weight parity base. 
See Fig.~\ref{fig:revisionfig31} for a flowchart of our algorithm.

\begin{figure}[htbp]
  \centering
    \includegraphics[width=10cm]{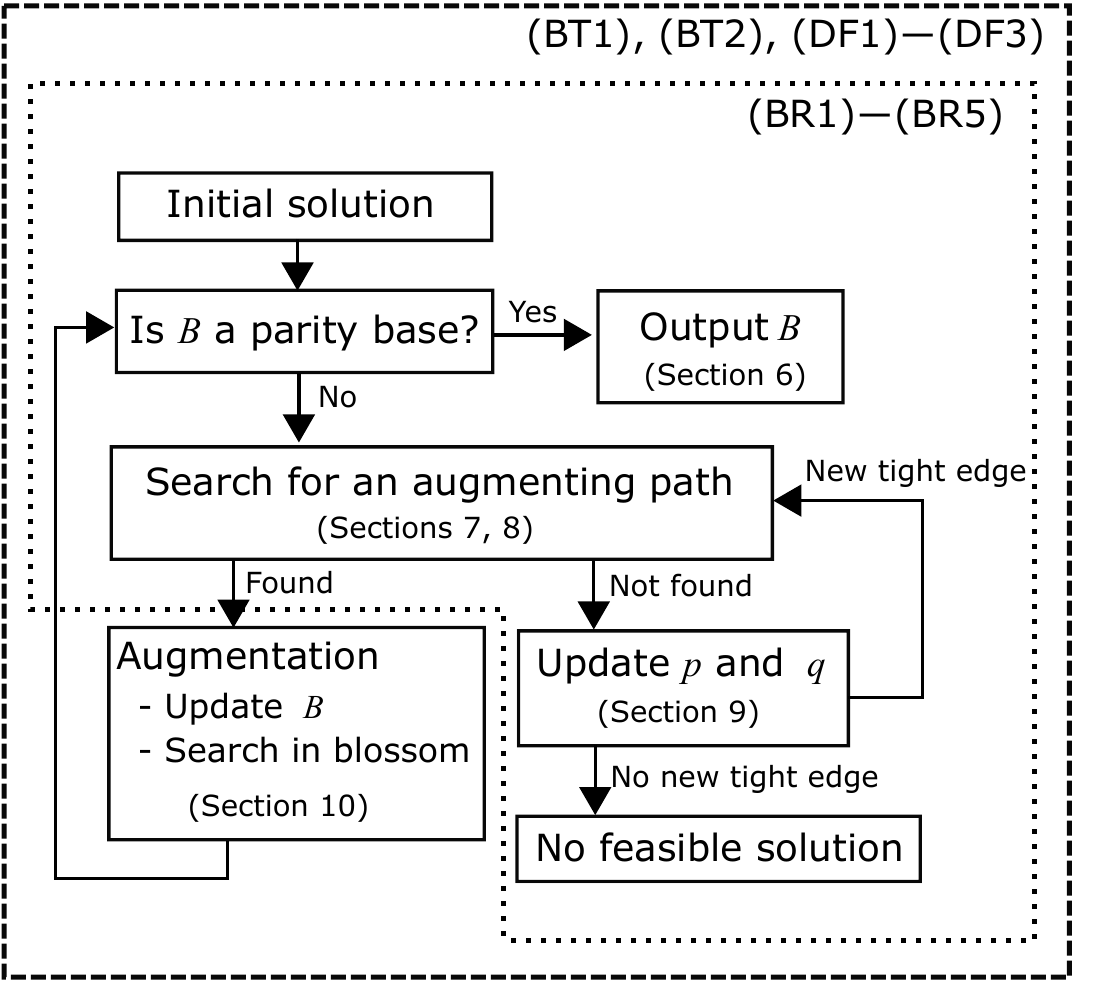}
      \caption{Flow chart of our algorithm. The conditions (BT1), (BT2), and (DF1)--(DF3) always hold, 
whereas (BR1)--(BR5) do not necessarily hold during the augmentation procedure in Section~\ref{sec:augmentation}.} 
    \label{fig:revisionfig31} 
\end{figure}

The rest of this paper is organized as follows. 

In Section~\ref{sec:blossoms}, 
we introduce new vertices and operations attached to blossoms. 
We describe some properties of blossoms kept in the algorithm, which we denote (BT1) and (BT2).

The feasibility of the dual variables is defined in Section~\ref{sec:dual}. 
The dual feasibility is denoted by (DF1)--(DF3). 
We also describe several properties of feasible dual variables
that are used in other sections.

In Section~\ref{sec:optimality}, 
we show that a parity base that admits feasible dual variables 
attains the minimum weight. 
The proof is based on the polynomial matrix formulation of 
the minimum-weight parity base problem given in Section~\ref{sec:problemdef}. 
Combining this with
some properties of the dual variables and 
the duality of the maximum-weight matching problem, 
we show the optimality of such a parity base. 

In Section~\ref{sec:search}, we describe a search procedure for an augmenting path. 
We first define an augmenting path, and then we describe our search procedure. 
Roughly, our procedure finds a part of the augmenting path outside the blossoms.  
The routing in each blossom is determined by a prescribed vertex set  
that satisfies some conditions, which we denote (BR1)--(BR5). 
Note that the search procedure may create new blossoms.

The validity of the procedure is shown in Section~\ref{sec:validity}. 
We show that the output of the procedure is an augmenting path
by using the properties (BR1)--(BR5) of the routing in each blossom. 
We also show that creating a new blossom does not violate 
the conditions (BT1), (BT2),  (DF1)--(DF3), and (BR1)--(BR5).

In Section~\ref{sec:dualupdatealgo}, 
we describe how to update the dual variables when 
the search procedure terminates without finding an augmenting path. 
We obtain new tight edges by updating the dual variables, and repeat the search procedure.  
We also show that if we cannot obtain new tight edges, then 
the instance has no feasible solution, i.e., there is no parity base.

If the search procedure succeeds in finding an augmenting path $P$, 
the algorithm updates the base $B$ along $P$. 
The details of this process are presented in Section~\ref{sec:augmentation}. 
Basically, we replace the base $B$ with the symmetric difference of $B$ and $P$. 
In addition, since there exist new vertices corresponding to the blossoms, 
we update them carefully to keep the conditions (BT1), (BT2), and (DF1)--(DF3). 
In order to define a new routing in each blossom, 
we apply the search procedure in each blossom, 
which enables us to keep the conditions (BR1)--(BR5).

Finally, in Section~\ref{sec:complexity}, 
we describe the entire algorithm and analyze its running time.  
We show that our algorithm solves the minimum-weight parity base problem 
in $O(n^3 m)$ time when $\K$ is a finite field of fixed order. 
When $\K = \mathbb{Q}$, it is not obvious that 
a direct application of our algorithm runs in polynomial time. 
However, we show that the minimum-weight parity base problem over $\mathbb Q$ can be solved 
in polynomial time by applying our algorithm over a sequence of finite fields.

\section{Blossoms}
\label{sec:blossoms}

In this section, we introduce buds and tips attached to blossoms
and construct auxiliary matrices that will be used in the definition of dual feasibility. 

Each blossom contains at most one source line. 
A blossom that contains a source line is called a {\em source blossom}. 
A blossom with no source line is called a {\em normal blossom}. 
Let $\Lambda_{\rm s}$ and $\Lambda_{\rm n}$ denote 
the sets of source blossoms and normal blossoms, respectively. 
Then, $\Lambda = \Lambda_{\rm s} \cup \Lambda_{\rm n}$. 
Let $\lambda$ denote the number of blossoms in $\Lambda$. 

Each normal blossom $H_i\in\Lambda_{\rm n}$ has a pair of associated 
vertices $b_i$ and $t_i$ outside $V$, which are called the {\em bud} and 
the {\em tip} of $H_i$, respectively. The pair $\{b_i,t_i\}$ is called a {\em dummy line}. 
To simplify the description, we denote $\bar b_i = t_i$ and $\bar t_i = b_i$.
The vertex set $V^*$ is defined by $V^* := V \cup T$ with $T:=\{b_i,t_i \mid H_i \in \Lambda_{\rm n}\}$. 
The tip $t_i$ is contained in $H_i$, whereas the bud $b_i$ is outside $H_i$. 
For every $i, j$ with $H_j \in \Lambda_{\rm n}$, we have 
$t_j\in H_i$ if and only if $H_j \subseteq H_i$. 
Similarly, we have $b_j\in H_i$ if and only if $H_j\subsetneq H_i$. 
Thus, each normal blossom $H_i$ is of odd cardinality. 
The algorithm keeps a subset $B^* \subseteq V^*$ such that $B^* \cap V=B$ 
and $|B^* \cap \{b_i, t_i\}| = 1$ for each $H_i \in \Lambda_{\rm n}$. 
It also keeps $H_i \cap V \not = H_j \cap V$ for distinct $H_i, H_j \in \Lambda$
and $H_i \cap V \not = \emptyset$ for each $H_i \in \Lambda$. 
This implies that $|\Lambda| = O(n)$, where $n = |V|$, and hence $|V^*| = O(n)$.

Recall that $U$ is the row set of $A$. 
The {\em fundamental cocircuit matrix} $C$ with respect to a base $B$ is 
a matrix with row set $B$ and column set $V\setminus B$ obtained 
by $C=A[U,B]^{-1}A[U,V\setminus B]$. 
In other words, $(I \ C)$ is obtained from $A$ 
by identifying $B$ and $U$, applying row transformations, 
and changing the ordering of columns. For a subset $S\subseteq V$, we have 
$B\triangle S\in\B$ if and only if $C[S]:=C[S\cap B,S\setminus B]$ is nonsingular.
Here, $\triangle$ denotes the symmetric difference. 
Then the following lemma characterizes the fundamental cocircuit matrix with respect 
to $B\triangle S$.
\begin{lemma}
\label{lem:pivot}
Suppose that $C$ is in the form of 
$C=\begin{pmatrix} \alpha & \beta \\ \gamma & \delta \end{pmatrix}$ 
with $\alpha=C[S]$ being nonsingular. Then 
\begin{equation*}
C':=\begin{pmatrix}
\alpha^{-1} & \alpha^{-1}\beta \\ 
-\gamma\alpha^{-1} & \delta-\gamma\alpha^{-1}\beta
\end{pmatrix}
\end{equation*}
is the fundamental cocircuit matrix with respect to $B\triangle S$. 
\end{lemma} 

\begin{proof}
In order to obtain the fundamental cocircuit matrix with respect to $B\triangle S$, 
we apply row elementary transformations 
to $(I \ C) = \begin{pmatrix} I & 0 & \alpha & \beta \\ 0 & I & \gamma & \delta \end{pmatrix}$
so that the columns corresponding to $B \triangle S$ form the identity matrix. 
Hence, the obtained matrix is 
$$
\begin{pmatrix} \alpha^{-1} & 0 \\ - \gamma \alpha^{-1} & I \end{pmatrix} 
\begin{pmatrix} I & 0 & \alpha & \beta \\ 0 & I & \gamma & \delta \end{pmatrix}
= \begin{pmatrix} \alpha^{-1} & 0 & I & \alpha^{-1} \beta \\ - \gamma \alpha^{-1} & I & 0 & \delta - \gamma \alpha^{-1} \beta \end{pmatrix},
$$
which shows that $C'$ is the fundamental cocircuit matrix with respect to $B\triangle S$. 
\end{proof}

This operation converting $C$ to $C'$ is called {\em pivoting around $S$}. 
We have the following property on the nonsingularity of their submatrices.

\begin{lemma}\label{lem:pivotsing}
Let $C$ and $C'$ be the fundamental cocircuit matrices with respect to $B$ and $B \triangle S$, respectively. 
Then, for any $X \subseteq V$, $C[X]$ is nonsingular if and only if $C'[X \triangle S]$ is nonsingular. 
\end{lemma}

\begin{proof}
Consider the matrix $(I \ C)$ whose column set is equal to $V$. 
Then, $C[X]$ is nonsingular if and only if the columns of $(I \ C)$ indexed by $X \triangle B$ form a nonsingular matrix. 
This is equivalent to that the corresponding columns of $(I \ C')$ form a nonsingular matrix, 
which means that $C'[X \triangle B \triangle (B \triangle S)] = C'[X \triangle S]$ is nonsingular. 
\end{proof}

The algorithm keeps a matrix $C^*$ whose row and column sets are 
$B^*$ and $V^* \setminus B^*$, respectively. The matrix $C^*$ is 
obtained from $C$ by attaching additional rows/columns
corresponding to $T$, and then pivoting around $T$. Thus we have $B^*\cap V=B$.  
In other words, the matrix obtained from $C^*$ by pivoting around $T$ contains $C$ as a submatrix (see (BT1) below).  
If the row and column sets of $C^*$ are clear, for a vertex set 
$X \subseteq V^*$, we denote $C^*[X]=C^*[X \cap B^*, X \setminus B^*]$. 

In our algorithm, 
the matrix $C^*$ satisfies the following properties. 
\begin{description}
\item[(BT1)] 
Let $C'$ be the matrix obtained from $C^*$ by pivoting around $T$. 
Then, $C'[V]$ is the fundamental cocircuit matrix with respect to $B = B^*\cap V$. 
\item[(BT2)] 
Each normal blossom $H_i \in \Lambda_{\rm n}$ satisfies the following.
\begin{itemize}
\item 
If $b_i\in B^*$ and $t_i \in V^* \setminus B^*$, then 
$C^*_{b_it_i} \neq 0$,
$C^*_{b_i v} = 0$ for any $v \in H_i\setminus B^*$ with $v\neq t_i$, and
$C^*_{u t_i} = 0$ for any $u \in B^* \setminus H_i$ with $u\neq b_i$ (see Fig.~\ref{fig:41}).  
\item
If $b_i\in V^* \setminus B^*$ and $t_i \in B^*$, then 
$C^*_{t_i b_i} \neq 0$,
$C^*_{u b_i} = 0$ for any $u \in B^*\cap H_i$ with $u\neq t_i$, and
$C^*_{t_i v} = 0$ for any $v \in (V^*\setminus B^*)\setminus H_i$ with $v\neq b_i$.
\end{itemize}
\end{description}

\begin{figure}[htbp]
 \begin{minipage}{0.5\hsize}
  \begin{center}
   \includegraphics[width=6cm]{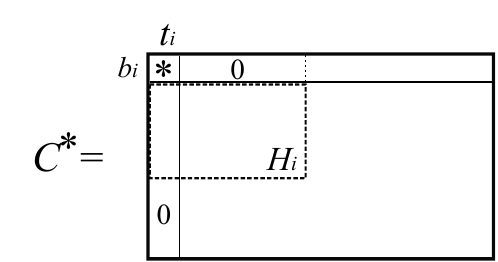}
  \end{center}
 \end{minipage}
 \begin{minipage}{0.5\hsize}
  \begin{center}
   \includegraphics[width=5cm]{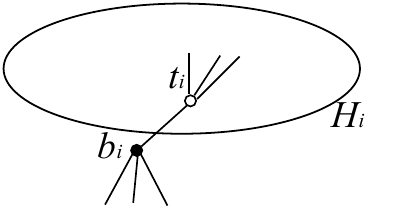}
  \end{center}
 \end{minipage}
    \caption{Illustration of (BT2).  In the right figure, real lines represent nonzero entries of $C^*$.} 
    \label{fig:41} 
\end{figure}

\section{Dual Feasibility}
\label{sec:dual}

In this section, we define feasibility of the dual variables 
and show their properties. 
Our algorithm for the minimum-weight parity base problem is designed so that 
it keeps the dual feasibility. 

Recall that a potential $p:V^*\to\R$, and a nonnegative variable $q:\Lambda\to\R_+$
are called dual variables. 
A blossom $H_i$ is said to be {\em positive} if $q(H_i) >0$. 
For distinct vertices $u, v \in V^*$ and for $H_i \in \Lambda$, 
we say that a pair $(u, v)$ {\em crosses} $H_i$ if $|\{u, v\} \cap H_i| =1$. 
For distinct $u, v \in V^*$, we denote by $I_{uv}$ the set of indices 
$i\in \{1,\ldots,|\Lambda|\}$ such that $(u, v)$ crosses $H_i$. 
We introduce the set $F^*$ of ordered vertex pairs defined by
$$
F^* := \{ (u, v) \mid u\in B^*,\,v\in V^*\setminus B^*,\, C^*_{uv}\neq 0\}.
$$ 
For distinct $u, v \in V^*$, we define
\[
Q_{uv} := \sum_{i \in I_{uv}} q(H_i). 
\]
The dual variables are called {\em feasible with respect to $C^*$ and $\Lambda$} if they satisfy the following.
\begin{description}
\item[(DF1)] $p(v)+p(\bar{v})=w_\ell$ for every line $\ell=\{v,\bar{v}\}\in L$.
\item[(DF2)] $p(v)-p(u)\geq Q_{uv}$ for every $(u,v) \in F^*$.
\item[(DF3)] $p(v)-p(u)=q(H_i)$ for every $H_i\in\Lambda_{\rm n}$ and $(u,v)\in F^*$ 
with $\{u,v\}=\{b_i,t_i\}$.   
\end{description}
If no confusion may arise, we omit $C^*$ and $\Lambda$ when we discuss dual feasibility. 

Note that if $\Lambda = \emptyset$, then 
$F^*$ corresponds to the nonzero entries of $C = C^*$, 
which shows that $(B \setminus \{u\}) \cup \{v\} \in \B$ holds for $(u,v) \in F^*$. 
This implies that (DF2) holds if 
$B\in\B$ is a base minimizing $p(B)=\sum_{u\in B}p(u)$, 
because $Q_{uv} = 0$ for any $(u,v) \in F^*$.  
We also note that (DF3) holds if $\Lambda = \emptyset$. 
Therefore, $p$ and $q$ are feasible if $p$ satisfies (DF1), $\Lambda = \emptyset$, and 
$B\in\B$ minimizes $p(B)=\sum_{u\in B}p(u)$ in $\B$.  
This ensures that the initial setting of the algorithm satisfies the dual feasibility.  

We now show some properties of feasible dual variables.

\begin{lemma}
\label{lem:keyodd}
Suppose that $p$ and $q$ are feasible dual variables. 
Let $X \subseteq V^*$ be a vertex subset such that $C^*[X]$ is nonsingular.
Then, we have 
$$ p(X \setminus B^*) - p(X \cap B^*) \geq 
\sum \{ q(H_i) \mid H_i \in \Lambda,\ \mbox{$|X \cap H_i|$ is odd}\}.$$
\end{lemma}
\begin{proof} 
Since $C^* [X]$ is nonsingular, there exists a perfect matching 
$M = \{(u_j,v_j)\mid j=1,\ldots,\mu\}$ between $X \cap B^*$ and $X \setminus B^*$ 
such that $u_j \in X \cap B^*$, $v_j \in X \setminus B^*$, and $C^*_{u_j v_j} \not= 0$ for $j=1,\dots,\mu$. 
The dual feasibility implies that $p(v_j)-p(u_j)\geq Q_{u_j v_j}$ for $j=1, \dots , \mu$. 
Combining these inequalities, we obtain 
\begin{equation}\label{eq:dual03} 
p(X \setminus B^*)-p(X \cap B^*) 
             \ge \sum_{j=1}^\mu Q_{u_j v_j} = \sum_{j=1}^\mu 
             \sum_{i \in I_{u_j v_j}} q(H_i). 
\end{equation}
If $|X \cap H_i|$ is odd, there exists an index $j$ such that 
$i \in I_{u_{j} v_{j}}$, which shows that the coefficient of $q(H_i)$ 
in the right hand side of (\ref{eq:dual03}) is at least $1$. 
This completes the proof
\end{proof}

We now consider the tightness of the inequality in Lemma~\ref{lem:keyodd}. 
Let $G^* = (V^*, F^*)$ be the undirected graph with vertex set $V^*$ and edge set $F^*$, where we regard $F^*$ as a set of unordered pairs. 
An edge $(u, v) \in F^*$ with $u \in B^*$ and $v \in V^* \setminus B^*$ 
is said to be {\em tight} if $p(v)-p(u) = Q_{uv}$. 
We say that a matching $M \subseteq F^*$ is {\em consistent with a blossom $H_i \in \Lambda$}
if at most one edge in $M$ crosses $H_i$.  
We say that a matching $M \subseteq F^*$ is {\em tight} if 
every edge of $M$ is tight and $M$ is consistent with every positive blossom $H_i$. 
As the proof of Lemma~\ref{lem:keyodd} clarifies, if there exists a tight perfect matching $M$ in 
the subgraph $G^*[X]$ of $G^*$ induced by $X$, then the inequality of Lemma~\ref{lem:keyodd} is tight. 
Furthermore, in such a case, every perfect matching in $G^*[X]$ must be tight, 
which is stated as follows.

\begin{lemma}
\label{lem:tightmatching}
For a vertex set $X \subseteq V^*$, if $G^*[X]$ has a tight perfect matching, then 
any perfect matching in $G^*[X]$ is tight.  
\end{lemma}

When $q(H_i)=0$ for some $H_i\in\Lambda$, one can delete $H_i$ from $\Lambda$ without 
violating the dual feasibility. In fact, removing such a source blossom does not affect 
the dual feasibility, (BT1), and (BT2). If $H_i$ is a normal blossom, 
then apply the pivoting operation around $\{b_i,t_i\}$ to $C^*$, remove $b_i$ and $t_i$ 
from $V^*$, and remove $H_i$ from $\Lambda$. This process is referred to as $\Expand(H_i)$. 
\begin{lemma}
\label{lem:expand}
If $q(H_i)=0$ for some $H_i\in\Lambda_{\rm n}$, 
the dual variables $(p,q)$ remain feasible 
and {\em (BT1)} and {\em (BT2)} hold
after $\Expand(H_i)$ is executed. 
\end{lemma}

\begin{proof}
We only consider the case when $b_i \in B^*$ and $t_i \in V^* \setminus B^*$, 
since we can deal with the case of $b_i \in V^* \setminus B^*$ and $t_i \in B^*$ in the same way. 
Let $C^*$ be the original matrix and $C'$ be the matrix obtained after $\Expand(H_i)$ is executed.  
Let $F^*$ (resp.~$F'$) be the ordered vertex pairs corresponding to the nonzero entries of $C^*$ (resp.~$C'$). 

Suppose that $p$ and $q$ are feasible with respect to $F^*$. 
In order to show that $p$ and $q$ are feasible with respect to $F'$, 
it suffices to consider (DF2), since (DF1) and (DF3) are obvious. 
Suppose that $(u, v) \in F'$. 
If $(u, v) \in F^*$, then $p(v) - p(u) \ge Q_{uv}$ by the dual feasibility with respect to $F^*$. 
Otherwise, we have $(u, v) \in F'$ and $(u, v) \not\in F^*$. 
By Lemma~\ref{lem:pivot}, 
$C'_{uv} = C^*_{uv} - C^*_{u t_i} (C^*_{b_i t_i})^{-1} C^*_{b_i v}$, 
and hence 
$(u, v) \in F'$ and $(u, v) \not\in F^*$ imply that $C^*_{b_i v} \neq 0$ and  $C^*_{u t_i} \neq 0$. 
Then, by the dual feasibility with respect to $F^*$, we obtain 
\begin{align*}
p(v) - p(b_i) &\ge Q_{b_i v}, \\
p(t_i) - p(u) &\ge Q_{u t_i}. 
\end{align*}
Furthermore, 
we have $p(b_i) = p(t_i)$ by (DF3) and 
$Q_{b_i v} + Q_{u t_i} = Q_{b_i v} + Q_{u t_i} + q(H_i) \ge Q_{uv}$. 
By combining these inequalities, 
we obtain $p(v) - p(u) \ge Q_{uv}$. 
This shows that (DF2) holds with respect to $F'$. 

By the definition of $\Expand(H_i)$, it is obvious that $C'$ satisfies (BT1). 

To show (BT2), let $H_j$ be a normal blossom that is different from $H_i$. 
Suppose that $b_j \in B^*$ and $t_j \in V^* \setminus B^*$.
we consider the following cases, separately. 
\begin{itemize}
\item
If $H_j \subseteq H_i$, then 
$C^*_{b_i v} = 0$ for any $v \in H_j \setminus B^*$. 
In particular, $C^*_{b_i t_j} = 0$. 
\item
If $H_i \subseteq H_j$, then 
$C^*_{u t_i} = 0$ for any $u \in B^* \setminus H_j$.
In particular, $C^*_{b_j t_i} = 0$. 
\item
If $H_i \cap H_j = \emptyset$, then 
we have that $C^*_{b_i t_j} = 0$ and $C^*_{b_j t_i} = 0$. 
\end{itemize}
In every case, we have that
$C'_{b_j v} = C^*_{b_j v} - C^*_{b_j t_i} (C^*_{b_i t_i})^{-1} C^*_{b_i v} = C^*_{b_j v}$ for any $v \in H_j \setminus B^*$, and 
$C'_{u t_j} = C^*_{u t_j} - C^*_{u t_i} (C^*_{b_i t_i})^{-1} C^*_{b_i t_j} = C^*_{u t_j}$ for any $u \in B^* \setminus H_j$.  
Therefore, 
$C'_{b_j t_j} = C^*_{b_j t_j} \not= 0$, 
$C'_{b_j v} = C^*_{b_j v} = 0$ for any $v \in H_j \setminus B^*$ with $v \not= t_j$, and 
$C'_{u t_j} = C^*_{u t_j} = 0$ for any $u \in B^* \setminus H_j$ with $u \not= b_j$. 
We can deal with the case when $b_j \in V^* \setminus B^*$ and $t_j \in B^*$ in a similar way. 
This shows that $C'$ satisfies (BT2).   
\end{proof}

\section{Optimality}
\label{sec:optimality}

In this section, we show that
if we obtain a parity base $B$ and feasible dual variables $p$ and $q$, then $B$ is a minimum-weight parity base. 

Note again that although $p$ and $q$ are called dual variables, they do not correspond to 
dual variables of an LP-relaxation of the minimum-weight parity base problem. 
Our optimality proof is based on the algebraic formulation of the problem (Lemma~\ref{lem:Pf}) and the duality of the maximum-weight matching problem.

\begin{theorem}\label{lem:optimality}
If $B := B^*\cap V$ is a parity base and there exist feasible dual variables $p$ and $q$, then $B$ is a minimum-weight parity base. 
\end{theorem}

\begin{proof}
Since the optimal value of the minimum-weight parity base problem is represented with $\deg_\theta\,\Pf\Phi_A(\theta)$
as shown in Lemma~\ref{lem:Pf}, 
we evaluate the value of $\deg_\theta\,\Pf\Phi_A(\theta)$, 
assuming that we have a parity base $B$ and feasible dual variables $p$ and $q$. 

Recall that $A$ is transformed to $(I \ C)$ by applying row transformations and column permutations, 
where $C$ is the fundamental cocircuit matrix with respect to the base $B$ obtained by $C=A[U,B]^{-1}A[U,V\setminus B]$. 
Note that the identity submatrix gives a one to one correspondence between $U$ and $B$, and 
the row set of $C$ can be regarded as $U$. 
We now apply the same row transformations and column permutations to $\Phi_A(\theta)$, 
and then apply also the corresponding column transformations and row permutations
to obtain a skew-symmetric polynomial matrix $\Phi_A'(\theta)$, that is, 
$$
\Phi_A'(\theta)=\left( 
\begin{array}{c|cc}
O &  I & C \\ \hline
-I  &    &   \\
-C^\top &  \multicolumn{2}{c}{\raisebox{5pt}[0pt][0pt]{\large $D'(\theta)$}} 
\end{array}\right) 
\begin{array}{l}
\leftarrow U \\ 
\leftarrow B \\
\leftarrow V\setminus B,  
\end{array}
$$
where $D'(\theta)$ is a skew-symmetric matrix obtained from $D(\theta)$ 
by applying row and column permutations simultaneously. 
Note that $\Pf\Phi_A'(\theta) = \pm \Pf\Phi_A(\theta)/\det A[U,B]$, where 
the sign is determined by the ordering of $V$. 

We now consider the following skew-symmetric matrix: 
$$\Phi_A^*(\theta)=
\left( 
\begin{array}{c|c|c|c|c}
\multicolumn{2}{c|}{}  &  O & \multicolumn{2}{c}{}   \\ \cline{3-3}
\multicolumn{2}{c|}{\raisebox{7pt}[0pt][0pt]{$O$}} & I  & \multicolumn{2}{c}{\raisebox{7pt}[0pt][0pt]{\quad $C^*$}\quad}   \\ \hline
 O & -I  & \multicolumn{2}{c|}{} & \\ \cline{1-2}
\multicolumn{2}{c|}{} & \multicolumn{2}{c|}{\raisebox{7pt}[0pt][0pt]{$D'(\theta)$}} & \raisebox{7pt}[0pt][0pt]{$O$} \\ \cline{3-5}
\multicolumn{2}{c|}{\raisebox{7pt}[0pt][0pt]{$-{C^*}^\top$}} & \multicolumn{2}{c|}{O} & O  
\end{array}
\right)
\begin{array}{l}
 \\
\raisebox{7pt}[0pt][0pt]{$\leftarrow U^*$ (identified with $B^*$)} \\
\leftarrow B \\
\leftarrow V\setminus B \\ 
\leftarrow T\setminus B^*.
\end{array}
$$
Here, the row and column sets of $\Phi_A^*(\theta)$ are both indexed by $W^* := U^* \cup V \cup (T \setminus B^*)$, 
where $U^*$ is the row set of $C^*$, which can be identified with $B^*$. 
Then, we have the following claim. 

\begin{claim}
It holds that $\deg_\theta\Pf\Phi^*_A(\theta)= \deg_\theta\Pf\Phi'_A(\theta) = \deg_\theta\Pf\Phi_A(\theta)$.
\end{claim}

\begin{proof}
Since $C^*$ satisfies (BT1), 
we can obtain 
$
\left( 
\begin{array}{c|c|c}
O & X & I   \\ \hline
I  & C & O 
\end{array}
\right)
$
from 
$
\left( 
\begin{array}{c|c|c}
O & \multicolumn{2}{c}{}  \\ \cline{1-1}
I & \multicolumn{2}{c}{\raisebox{7pt}[0pt][0pt]{\ $C^*$}}  
\end{array}
\right)
$
by applying elementary row transformations, where $X$ is some matrix. 
Here, the row and column sets are 
$U^*$ and $B \cup (V \setminus B) \cup (T \setminus B^*)$, respectively. 
We apply the same row transformations and their corresponding column transformations to $\Phi_A^*(\theta)$. 
Then, we obtain the following matrix: 
$$\hat \Phi_A(\theta)=
\left( 
\begin{array}{c|c|c|c|c}
\multicolumn{2}{c|}{}  &  O & X & I    \\ \cline{3-5}
\multicolumn{2}{c|}{\raisebox{7pt}[0pt][0pt]{$O$}} & I  & C & O   \\ \hline
 O & -I  & \multicolumn{2}{c|}{} & \\ \cline{1-2}
- X^\top & - C^\top & \multicolumn{2}{c|}{\raisebox{7pt}[0pt][0pt]{$D'(\theta)$}} & \raisebox{7pt}[0pt][0pt]{$O$} \\ \hline
- I & O & \multicolumn{2}{c|}{O} & O  
\end{array}
\right)
\begin{array}{l}
 \\
\raisebox{7pt}[0pt][0pt]{$\leftarrow U^*$ (identified with $B^*$)} \\
\leftarrow B \\
\leftarrow V\setminus B \\ 
\leftarrow T\setminus B^*, 
\end{array}
$$
and hence 
$\deg_\theta\Pf\Phi^*_A(\theta) = \deg_\theta\Pf \hat \Phi_A(\theta)$. 
Since $\Pf \hat \Phi_A(\theta) = \pm \Pf \Phi'_A(\theta)$, 
we have that 
$$
\deg_\theta\Pf\Phi^*_A(\theta) = 
\deg_\theta\Pf \hat \Phi_A(\theta) = 
\deg_\theta\Pf \Phi'_A(\theta) = 
\deg_\theta\Pf\Phi_A(\theta), 
$$
which completes the proof. 
\end{proof}

In what follows, we evaluate $\deg_\theta\Pf\Phi^*_A(\theta)$. 
Construct a graph $\Gamma^*=(W^*,E^*)$ 
with edge set $E^*:=\{(u,v) \mid (\Phi^*_A(\theta))_{u v} \neq 0 \}$.
Each edge $(u,v) \in E^*$ has a weight $\deg_\theta\,(\Phi_A^*(\theta))_{u v}$. 
Then it can be easily seen by the definition of Pfaffian that the maximum weight of a perfect matching in 
$\Gamma^*$ is at least $\deg_\theta \Pf\Phi^*_A(\theta) = \deg_\theta \Pf \Phi_A(\theta)$.  
Let us recall that the dual linear program of the maximum-weight perfect matching problem 
on $\Gamma^*$ is formulated as follows. 
\begin{eqnarray}
\mbox{Minimize} & & \sum_{v\in W^*}\pi(v)-\sum_{Z\in\Omega}\xi(Z) \notag \\
\mbox{subject to} & & \pi(u)+\pi(v)-\sum_{Z \in \Omega_{uv}} \xi(Z) \geq \deg_\theta\,(\Phi_A^*(\theta))_{u v} 
\quad\quad (\forall (u,v)\in E^*), \label{eq:dualmatching}\\
                  & & \xi(Z)\geq 0 \quad\quad\quad (\forall Z\in\Omega), \notag
\end{eqnarray}
where $\Omega=\{Z\mid Z\subseteq W^*,\,\mbox{$|Z|$: odd}, |Z|\geq 3\}$ 
and $\Omega_{uv}=\{Z\mid Z\in\Omega,|Z \cap \{u, v\}| = 1\}$  (see, e.g., \cite[Theorem 25.1]{Sch03}).
In what follows, we construct a feasible solution $(\pi, \xi)$ of this linear program.
The objective value provides an upper bound on the maximum weight of a perfect matching 
in $\Gamma^*$, and consequently serves as an upper bound on $\deg_\theta \Pf\Phi_A(\theta)$. 

Since 
$U^*$ can be identified with $B^*$, 
we can naturally define a bijection $\eta: B^* \to U^*$ between $B^*$ and $U^*$. 
We define $\pi: W^*\to \R$ by
$$\pi(v) = 
\begin{cases}
p(v) & \mbox{if $v \in V \cup (T \setminus B^*)$}, \\
-p(\eta^{-1}(v)) & \mbox{if $v \in U^*$},
\end{cases}
$$
For each $i \in \{1, \dots , \lambda\}$, we introduce  
$Z_i = (H_i \setminus (T \cap B^*)) \cup \eta(H_i \cap B^*) \subseteq W^*$
and set $\xi(Z_i) = q(H_i)$ (see Fig.~\ref{fig:51}). 
Since $H_i$ is of odd cardinality and there is no source line in $G$, 
we see that 
$$|Z_i|=|H_i\setminus (T \cap B^*)|+|H_i\cap B^*|=|H_i|+|H_i\cap B|$$  
is odd and $|Z_i|\geq 3$. 
Define $\xi(Z) = 0$ for any $Z \in \Omega \setminus \{Z_1, \dots , Z_\lambda\}$. 
We now show the following claim. 

 \begin{figure}[htbp]
   \centering
    \includegraphics[width=6cm,clip]{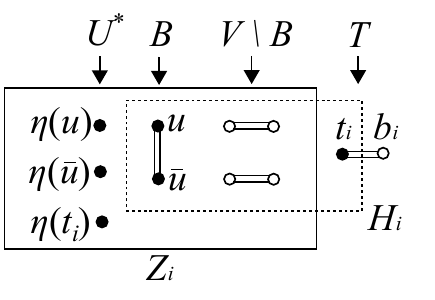}
    \caption{Definition of $Z_i$. Lines and dummy lines are represented by double bonds.} 
     \label{fig:51} 
 \end{figure}

\begin{claim}
The dual variables $\pi$ and $\xi$ defined as above form a feasible solution of the linear program (\ref{eq:dualmatching}).  
\end{claim}

\begin{proof}
Suppose that $(u,v)\in E^*$. 
If $u, v \in V$ and $u = \bar v$, then  
(DF1) shows that $\pi(u) + \pi(v) = p(\bar v) + p(v) = w_\ell = \deg_\theta\,(\Phi_A^*(\theta))_{u v}$, 
where $\ell = \{v, \bar v\}$. 
Since $|Z_i \cap \{v, \bar v\}|$ is even for any $i \in \{1, \dots , \lambda\}$, this shows (\ref{eq:dualmatching}).
If $u \in U$ and $v \in B$, then $(u, v) \in E^*$ implies that $u = \eta(v)$, 
and hence $\pi(u) + \pi(v) = 0$, 
which shows (\ref{eq:dualmatching}) 
as $|Z_i \cap \{u, v\}|$ is even for any $i \in \{1, \dots , \lambda\}$.

The remaining case of $(u,v)\in E^*$ is when 
$u \in U^*$ and $v \in V^*\setminus B^*$.  
That is, it suffices to show that $(u, v)$ satisfies (\ref{eq:dualmatching}) 
if $C^*_{uv} \neq 0$. 
By the definition of $\pi$, we have $\pi(u)+\pi(v) = p(v) - p(u')$, 
where $u'=\eta^{-1}(u)$. By the definition of $Z_i$, we have 
$Z_i\in\Omega_{uv}$ if and only if $i \in I_{u'v}$, which shows that 
$$\sum_{i:\,Z_i\in\Omega_{uv}} \xi(Z_i) = \sum_{i \in I_{u'v}} q(H_i).$$
Since $C^*_{uv} \not= 0$, by (DF2),  we have 
$$p(v) - p(u') \ge Q_{u'v}=\sum_{i \in I_{u'v}} q(H_i).$$  
Thus, we obtain 
$$\pi(u)+\pi(v) - \sum_{i:\, Z_i\in\Omega_{uv}} \xi(Z_i) \ge 0,$$
which shows that $(u, v)$ satisfies (\ref{eq:dualmatching}). 
\end{proof}

The objective value of this feasible solution is 
\begin{eqnarray}\label{eq:opt04}\nonumber
\sum_{v\in W^*}\pi(v)-\sum_{i=1}^\lambda \xi(Z_i) &  
= & \sum_{v\in V \setminus B} p(v) + \sum_{v\in T\setminus B^*}p(v)-\sum_{v\in T\cap B^*}p(v)
-\sum_{i=1}^\lambda\xi(Z_i) \\ 
& = & \sum_{v\in V \setminus B} p(v) = \sum_{\ell\subseteq V\setminus B} w_\ell,
\end{eqnarray}
where the first equality follows from the definition of $\pi$, 
the second one follows from the definition of $\xi$ and (DF3), 
and the third one follows from (DF1). 
By the weak duality of the maximum-weight matching problem, we have 
\begin{align}
\sum_{v\in W^*}\pi(v)-\sum_{i=1}^\lambda \xi(Z_i) 
&\ge 
\mbox{(maximum weight of a perfect matching in $\Gamma^*$)} \notag \\ 
&\ge 
\deg_\theta \Pf\Phi_A^*(\theta) = \deg_\theta \Pf\Phi_A(\theta). \label{eq:opt05}
\end{align} 
On the other hand, Lemma~\ref{lem:Pf} shows that 
any parity base $B'$ satisfies that 
\begin{equation}\label{eq:opt06} 
\sum_{\ell \subseteq B'} w_\ell \ge \sum_{\ell \in L} w_\ell - \deg_\theta \Pf\Phi_A(\theta),
\end{equation} 
Combining (\ref{eq:opt04})--(\ref{eq:opt06}), we have 
$\sum_{\ell\subseteq V\setminus B} w_\ell = \deg_\theta \Pf\Phi_A(\theta)$, which means 
$B$ is a minimum-weight parity base by Lemma~\ref{lem:Pf}. 
\end{proof}

\section{Finding an Augmenting Path}
\label{sec:search}

In this section, we define an augmenting path and present a procedure for finding one. 
The validity of our procedure is shown in Section~\ref{sec:validity}. 

Suppose we are given $V^*$, $B^*$, $C^*$, $\Lambda$, and feasible dual 
variables $p$ and $q$. Let $F^\circ \subseteq F^*$ be the set of 
tight edges, i.e., 
$F^\circ = \{ (u, v) \in F^* \mid u \in B^*,\ v \in V^* \setminus B^*,\ p(v) - p(u) = Q_{uv}\}$. 
Our procedure works primarily on the undirected graph $G^\circ=(V^*, F^\circ)$, 
where we ignore the ordering of the vertices when we regard $F^\circ$ or $F^*$ as an edge set. 
For a vertex set $X \subseteq V^*$, $G^\circ[X]$ denotes the subgraph of $G^\circ$ induced by $X$. 
For $H_i \in \Lambda$, 
define $H^-_i$ as 
$$
H^-_i  = \{ v \in H_i \setminus \{t_i\} \mid \mbox{there is an edge in $F^*$ between $v$ and $V^* \setminus H_i$} \}. 
$$
Here, $\{t_i\}$ is regarded as $\emptyset$ if $H_i \in \Lambda_{\rm s}$. 
This definition shows that we can ignore $H_i \setminus H^-_i$ when we consider edges in $F^*$ (or $F^\circ$) 
connecting $H_i$ and $V^* \setminus H_i$.

Roughly, our procedure finds a part of the augmenting path outside the blossoms.  
The routing in each blossom $H_i$ is determined by a prescribed vertex set $R_{H_i}(x)$ for $x \in H^\bullet_i$, 
where $H^\bullet_i := H^-_i \cup  (H_i \cap V)$.
For any $i \in \{1, \dots , \lambda\}$ and for any $x \in H^\bullet_i$, 
the prescribed vertex set $R_{H_i}(x) \subseteq H_i$ is assumed to satisfy the following. 
\begin{description}
\item[(BR1)]
$x \in R_{H_i}(x) \subseteq H_i$. 
\item[(BR2)]
If $H_i \in \Lambda_{\rm n}$, then $R_{H_i}(x)$ consists of lines, dummy lines, and the tip $t_i$.  
If $H_i \in \Lambda_{\rm s}$, then $R_{H_i}(x)$ consists of lines, dummy lines, and a source vertex. 
\item[(BR3)]
For any $H_j \in \Lambda_{\rm n}$ with $R_{H_i}(x) \cap H_j \not= \emptyset$ and $H_j \subsetneq H_i$,
it holds that $\{b_j, t_j \} \subseteq R_{H_i}(x)$.
\end{description}
See Fig.~\ref{fig:revisionfig04} for an example of $R_{H_i}(x)$.
We sometimes regard $R_{H_i}(x)$ as a sequence of vertices, and in such a case, 
the last two vertices are $\bar x x$. 
We also suppose that the first vertex of $R_{H_i}(x)$ is $t_i$ if $H_i \in \Lambda_{\rm n}$ 
and the unique source vertex in $R_{H_i}(x)$ if $H_i \in \Lambda_{\rm s}$.  
Each blossom $H_i \in \Lambda$ is assigned a total order $<_{H_i}$ 
among all the vertices in $H^\bullet_i$. In the procedure, $R_{H_i}(x)$ keeps 
additional properties which will be described in Section~\ref{sec:propertyR}. 

\begin{figure}[htbp]
  \centering
    \includegraphics[width=7cm]{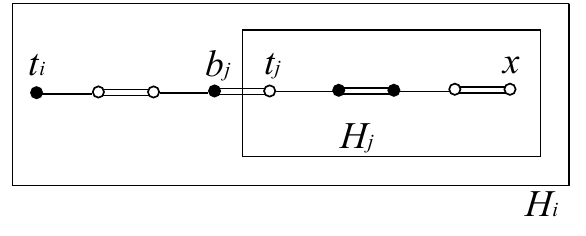}
\vspace*{-0.4cm}
      \caption{An example of $R_{H_i}(x)$.} 
    \label{fig:revisionfig04} 
\end{figure}

We say that a vertex set $P \subseteq V^*$ is an augmenting path
if it satisfies the following properties. 
\begin{description}
\item[(AP1)]
$P$ consists of normal lines, dummy lines, and two vertices from distinct source lines. 
\item[(AP2)]
For each $H_i \in \Lambda$, either $P \cap H_i = \emptyset$ or $P \cap H_i = R_{H_i}(x_i)$ 
for some $x_i \in H^\bullet_i$. 
\item[(AP3)]
$G^\circ[P]$ has a unique tight perfect matching. 
\end{description}

By (AP1), (AP2), and (BR2), we have the following observation. 
\begin{observation}
\label{obs:budtipinP}
For an augmenting path $P$ and for each $H_i \in \Lambda_{\rm n}$ with $P \cap H_i \not= \emptyset$, 
it holds that $\{b_i, t_i\} \subseteq P$. 
\end{observation}

In the rest of this section, we describe how to find an augmenting path. 
Section~\ref{sec:searchproc} is devoted to the search procedure, which 
calls two procedures: $\Blossom$ and $\DBlossom$. 
The details of these procedures are described in Sections~\ref{sec:createblossom} and~\ref{sec:graft}, respectively. 
In Section~\ref{sec:basicpro}, 
we show that the procedure keeps some conditions.

\subsection{Search Procedure}
\label{sec:searchproc}

In this subsection, we describe a procedure for searching for an augmenting path. 
The procedure performs the breadth-first search using a queue to grow paths from 
source vertices. A vertex $v\in V^*$ is labeled 
and put into the queue when it is reached by the search. 
The procedure picks the first labeled element from the queue, 
and examines its neighbors. 
A linear order $\prec$ is defined on the labeled vertex set so that 
$u \prec v$ means $u$ is labeled prior to $v$. 

For each $x\in V^*$, 
we define $K(x) = H_i \cup \{b_i\}$ if there exists a maximal blossom $H_i$ such that 
$H_i$ is a normal blossom with $x \in H_i \cup \{b_i\}$, and 
define $K(x) = H_i$ if there exists a maximal blossom $H_i$ such that 
$H_i$ is a source blossom with $x \in H_i$. 
If such a blossom does not exist, then 
it is called {\em single} and we denote $K(x) = \{x, \bar x\}$. 
The procedure also labels some blossoms with $\oplus$ or $\ominus$, 
which will be used later for modifying dual variables.
With each labeled vertex $v$, the procedure associates a path $P(v)$ and its subpath $J(v)$,  
where a path is a sequence of vertices. 
The first vertex of $P(v)$ is a 
labeled vertex in a source line and the last one is $v$. 
The reverse path of $P(v)$ is denoted by $\overline{P(v)}$. 
For a path $P(v)$ and a vertex $r$ in $P(v)$, we denote by $P(v|r)$ 
the subsequence of $P(v)$ after $r$ (not including $r$). 
We sometimes identify a path with its vertex set. 
When an unlabeled vertex $u$ is examined in the procedure, 
we assign a vertex $\rho(u)$ and a path $I(u)$. 
Roughly, $\rho(u)$ is a neighbor of $u$ such that  
$u$ is examined when we pick up $\rho(u)$ from the queue.  
Paths $I(u)$ and $J(v)$, where $u$ is an unlabeled vertex and $v$ is a labeled vertex, 
are used to decompose a search path as we will see in Lemma~\ref{lem:dec} later. 
Roughly, $I(u)$ and $J(v)$ represent ``fractions'' of the search path containing $u$ and $v$, respectively.
The procedure is described as follows.

\begin{description}
\item[Procedure] $\Search$ 
\item[Step 0:]
Initialize the objects so that 
the queue is empty, every vertex is unlabeled, and 
every blossom is unlabeled.  

\item[Step 1:] 
While there exists an unlabeled single vertex $x$ in a source line, 
label $x$ with $P(x):= J(x) := x$ and
put $x$ into the queue. 
While there exists a source line $\{x, \bar x\}$ such that $K(x)=K(\bar x)=\{x, \bar x\}$ and $x$ is adjacent to $\bar x$ in $G^\circ$, 
add a new source blossom $H=\{x, \bar x\}$ to $\Lambda$, label $H$ with $\oplus$, 
and define $R_H(x) := x$ and $R_H(\bar x) := \bar x$. 
While there exists an unlabeled maximal source blossom $H_i \in \Lambda_{\rm s}$, 
label $H_i$ with $\oplus$ and do the following: 
for each vertex $x \in H^\bullet_i$ in the order of $<_{H_i}$, 
label $x$ with $P(x):=J(x):=R_{H_i}(x)$ and put $x$ into the queue. 

\item[Step 2:] If the queue is empty, then 
return $\emptyset$ and terminate the procedure (see Section~\ref{sec:dualupdatealgo}). 
Otherwise, remove the first element $v$ from the queue. 
\item[Step 3:] While there exists a labeled vertex $u$ adjacent to $v$ in $G^\circ$ 
with $K(u) \not= K(v)$, choose such $u$ that is minimum with respect to $\prec$ and 
do the following (3-1) and (3-2) (see Fig.~\ref{fig:revisionfig05}). 
\begin{description}
\item[(3-1)] If the first elements in $P(v)$ and in $P(u)$ belong to different source lines,
then return $P:= P(v) \overline{ P(u)}$ as an augmenting path. 
\item[(3-2)] Otherwise, apply $\Blossom(v,u)$  
to add a new blossom to $\Lambda$.  
\end{description}
\item[Step 4:] While there exists an unlabeled vertex $u$ adjacent to $v$ in $G^\circ$ 
with $K(u) \not= K(v)$
such that $\rho(u)$ is not assigned, do the following (4-1)--(4-3). 
\begin{description} 
\item[(4-1)] 
If $K(u) = \{u, \bar u\}$, then 
label $\bar{u}$ with $P(\bar{u}):=P(v)u\bar{u}$ and $J(\bar{u}) := \{\bar u\}$, 
set $\rho(u) := v$ and $I(u) := \{u\}$, and put $\bar{u}$ into the queue (see Fig.~\ref{fig:revisionfig06}). 
Furthermore, if $(v, \bar u) \in F^\circ$, then apply $\Blossom(\bar u,v)$. 
\item[(4-2)] 
If $K(u) = H_i \cup \{b_i\}$ for some $H_i \in \Lambda_{\rm n}$ and 
$(v, b_i) \in F^\circ$, then 
apply $\DBlossom(v,H_i)$ (see Fig.~\ref{fig:revisionfig09}). 
\item[(4-3)] 
If $K(u) = H_i \cup \{b_i\}$ for some $H_i \in \Lambda_{\rm n}$ and  
$(v, b_i) \not\in F^\circ$, 
then choose $y \in H^\bullet_i$ with $(v, y) \in F^\circ$
that is minimum with respect to $<_{H_i}$, and do the following.\footnote{Such $y$ always exists, because $u$ satisfies the condition.}  
Label $H_i$ with $\ominus$, 
label $b_i$ with $P(b_i):=P(v) \overline{R_{H_i}(y)}b_i$ and $J(b_i) := \{b_i\}$, 
and put $b_i$ into the queue. For each unlabeled vertex $x \in H^\bullet_i$, 
set $\rho(x) := v$ and $I(x) := \overline{R_{H_i}(x)}$ (see Fig.~\ref{fig:revisionfig10}).  
\end{description}
\item[Step 5:] Go back to Step 2.  
\end{description}

\begin{figure}[htbp]
  \centering
    \includegraphics[width=11cm]{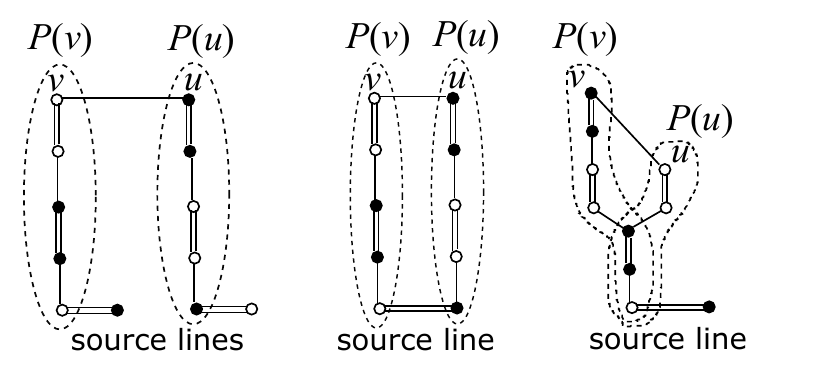}
      \caption{Illustrations of Step 3. We apply (3.1) for the leftmost case, and apply (3.2) for the other cases.}
    \label{fig:revisionfig05} 
\end{figure}

\begin{figure}[htbp]
 \begin{minipage}{0.32\hsize}
  \begin{center}
   \includegraphics[width=21mm]{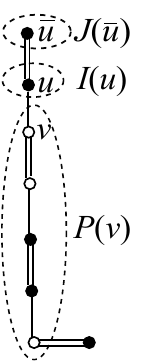}
  \end{center}
\vspace{-12pt}
  \caption{Step (4-1).}
  \label{fig:revisionfig06}
 \end{minipage}
 \begin{minipage}{0.33\hsize}
  \begin{center}
   \includegraphics[width=33mm]{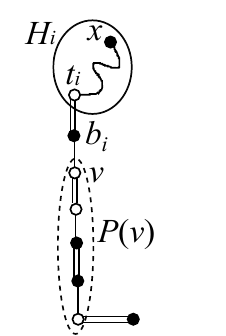}
  \caption{Step (4-2).}
  \label{fig:revisionfig09}
  \end{center}
 \end{minipage}
 \begin{minipage}{0.32\hsize}
  \begin{center}
   \includegraphics[width=45mm]{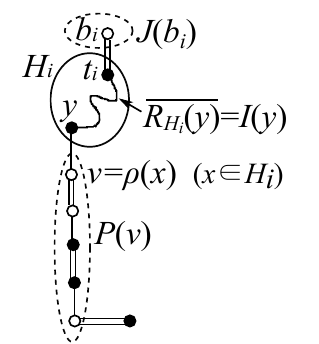}
  \caption{Step (4-3).}
  \label{fig:revisionfig10}
  \end{center}
 \end{minipage}
\end{figure}

\subsection{Creating a Blossom}
\label{sec:createblossom}

In this subsection, we describe procedure $\Blossom$ that creates a new blossom, 
which is called in Steps (3-2) and (4-1) of $\Search$.

\begin{description}
\item[Procedure] $\Blossom(v,u)$ 
\item[Step 1:] 
Let $c$ be the last vertex in $P(v)$ such that $K(c)$ contains a vertex in $P(u)$. 
Let $d$ be the last vertex in $P(u)$ contained in $K(c)$. 
Note that $K(c) = K(d)$. 
If $c=d$, then define $H := \bigcup \{K (x) \mid x \in P(v|c) \cup P(u|d) \}$ and $r := c$.
If $c\not=d$, then define $H := \bigcup \{K (x) \mid x \in P(v|c) \cup P(u|d) \cup \{c\} \}$ 
and let $r$ be the last vertex in $P(v)$ not contained in $H$ if exists. 
See Fig.~\ref{fig:61} for an example.

\begin{figure}[htbp]
  \centering
    \includegraphics[width=8cm]{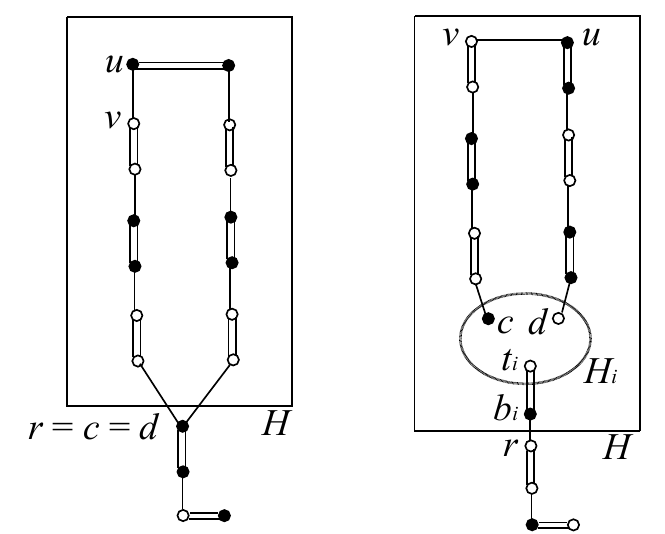}
    \caption{Definition of $H$.} 
    \label{fig:61} 
\end{figure}

\item[Step 2:]
If $H$ contains no source line, then define $g$ to be the vertex subsequent to $r$ in $P(v)$, 
introduce new vertices $b$ and $t$, namely $V^*:=V^*\cup\{b,t\}$, 
and add $t$ to $H$, namely $H:=H\cup\{t\}$. 
Update $B^*$, $C^*$, and $p$ as follows (see Fig.~\ref{fig:revisionfig13}). 
\begin{itemize}
\item 
If $r\in B^*$ and $g\in V^*\setminus B^*$, then  
$B^*:=B^*\cup\{b\}$, $C^*_{bt}:=C^*_{rg}$, $C^*_{by}:=C^*_{ry}$ for $y\in H\setminus B^*$, 
$C^*_{by}:=0$ for $y\in (V^*\setminus B^*)\setminus H$,  
$C^*_{xt}:=C^*_{xg}$ for $x\in B^*\setminus H$,  
$C^*_{xt}:=0$ for $x\in B^*\cap H$, and $p(b):=p(t):=p(r)+Q_{rb}$. 
\item 
If $r\in V^*\setminus B^*$ and $g\in B^*$, then  
$B^*:=B^*\cup\{t\}$, $C^*_{tb}:=C^*_{gr}$, $C^*_{xb}:=C^*_{xr}$ for $x\in B^*\cap H$,  
$C^*_{xb}:=0$ for $x\in B^*\setminus H$, 
$C^*_{ty}:=C^*_{gy}$ for $y\in (V^*\setminus B^*)\setminus H$, 
$C^*_{ty}:=0$ for $y\in H\setminus B^*$, and $p(b):=p(t):=p(r)-Q_{rb}$. 
\item Apply the pivoting operation around $\{b,t\}$ to $C^*$, namely $B^*:=B^*\triangle\{b,t\}$, 
and update $F^*$ accordingly. 
\end{itemize}

\begin{figure}[htbp]
  \centering
    \includegraphics[width=13cm]{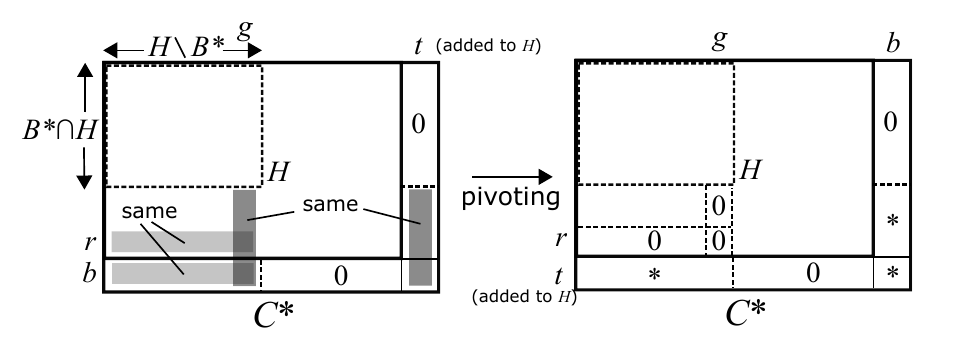}
      \caption{Definition of $C^*$. By the definition, $C^*_{ry}=0$ for $y \in H \setminus B^*$ and $C^*_{xg}=0$ for $x \in B^* \setminus H$ after the pivoting operation (see Lemma~\ref{lem:72}).} 
    \label{fig:revisionfig13} 
\end{figure}

\item[Step 3:]
If $H$ contains no source line, then for each labeled vertex $x$ with 
$P(x) \cap H\not= \emptyset$, replace $P(x)$ by $P(x) := P(r) b t P(x|r)$. 
Label $t$ with $P(t) := P(r)bt$ and $J (t) := \{t\}$, and
extend the ordering $\prec$ of the labeled vertices
so that $t$ is just after $r$, i.e., 
$r \prec t$ and no element is between $r$ and $t$. 
For each vertex $x \in H$ with $\rho(x) = r$, update $\rho(x)$ as $\rho(x) := t$.
Set $\rho(b) := r$ and $I(b) := \{b\}$ (see Fig.~\ref{fig:revisionfig11}).

\begin{figure}[htbp]
  \centering
    \includegraphics[width=8cm]{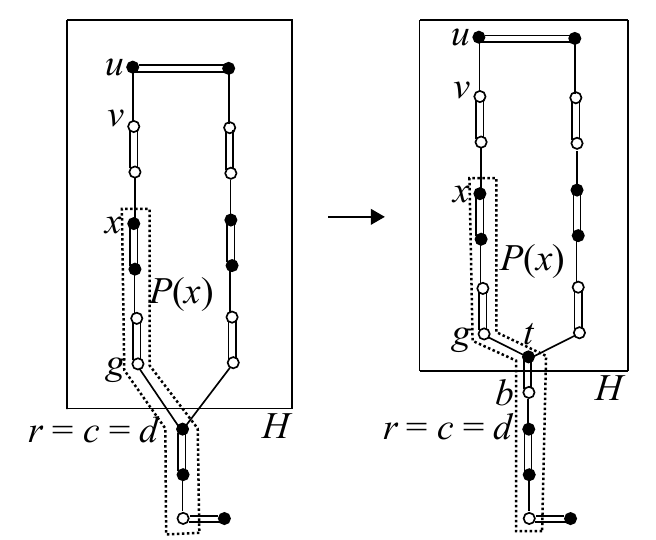}
    \caption{Update of $P(x)$.} 
    \label{fig:revisionfig11} 
\end{figure}

\item[Step 4:]
For each unlabeled vertex $x\in H^\bullet$, label $x$  as follows.
\begin{enumerate}
\item[(i)]
If $K(x) = \{x, \bar x\}$ and $x \in P(u|d)$, then $P(x):=P(v)\overline{P(u|x)}x$. 
\item[(ii)]
If $K(x) = \{x, \bar x\}$ and $x \in P(v|c)$, then $P(x):=P(u)\overline{P(v|x)}x$. 
\item[(iii)]
If $K(x) = H_i \cup \{b_i\}$ for some $H_i \in \Lambda_{\rm n}$ labeled with $\oplus$ such that $x=b_i$ and $x \in P(u|d)$, then $P(x):=P(v)\overline{P(u|x)}x$. 
\item[(iv)]
If $K(x) = H_i \cup \{b_i\}$ for some $H_i \in \Lambda_{\rm n}$ labeled with $\oplus$ such that $x=b_i$ and $x \in P(v|c)$, then $P(x):=P(u)\overline{P(v|x)}x$. 
\item[(v)]
If $K(x) = H_i \cup \{b_i\}$ for some $H_i \in \Lambda_{\rm n}$ labeled with $\ominus$ such that $x \in H^\bullet_i$ and $t_i \in P(u|d)$, then $P(x):=P(v)\overline{P(u|t_i)} R_{H_i}(x)$. 
\item[(vi)]
If $K(x) = H_i \cup \{b_i\}$ for some $H_i \in \Lambda_{\rm n}$ labeled with $\ominus$ such that $x \in H^\bullet_i$ and $t_i \in P(v|c)$, then $P(x):=P(u)\overline{P(v|t_i)} R_{H_i}(x)$. 
\end{enumerate}
Define $J(x) := P(x|t)$ and put $x$ into the queue (see Fig.~\ref{fig:revisionfig12}). 
Here, we choose the vertices in the ordering such that the following conditions hold. 
\begin{itemize}
\item
For two unlabeled vertices $x,y\in H^\bullet$, 
if $\rho(x) \succ \rho(y)$, then we choose $x$ prior to $y$. 
\item
For two unlabeled vertices $x,y\in H^\bullet$, 
if $\rho(x) = \rho(y)$, $K(x) = K(y) = H_i \cup \{b_i\}$, and $x <_{H_i} y$, then
we choose $x$ prior to $y$.
\item
If $r=c=d\neq u$ holds, then no element is chosen 
between $g$ and $h$, where $h$ is the vertex subsequent 
to $t$ in $P(u)$. Note that this condition makes 
sense only when $K(g)$ or $K(h)$ corresponds to 
a blossom labeled with $\ominus$.  
\end{itemize}

\begin{figure}[htbp]
  \centering
    \includegraphics[width=4.5cm]{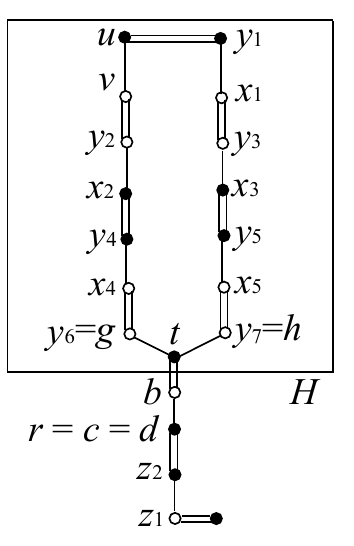}
    \caption{Illustration of Step 4 of $\Blossom(v,u)$. In this example, we define $P(y_3) = z_1 z_2 r b t y_6 x_4 y_4 x_2 y_2 v u y_1 x_1 y_3$. 
When $x_1 \succ x_2 \succ \dots \succ x_5 \succ t$, we choose $y_1, y_2, \dots , y_5, y_6, y_7$ (or $y_1, y_2, \dots , y_5, y_7, y_6$) in this order. 
} 
    \label{fig:revisionfig12} 
\end{figure}

\item[Step 5:]
Label $H$ with $\oplus$. 
Define $R_{H} (x) := P(x|b)$ for each $x \in H^\bullet$ 
if $H$ contains no source line, and define  
$R_{H} (x) := P(x)$ for each $x \in H^\bullet$ 
if $H$ contains a source line.   
Define $<_{H}$ by the ordering $\prec$ of the labeled vertices in $H^\bullet$.  
Add $H$ to $\Lambda$ with $q(H)=0$ regarding $b$ and $t$, if exist, 
as the bud and the tip of $H$, respectively, 
and update $\Lambda_{\rm n}$, 
$\Lambda_{\rm s}$, $\lambda$, $G^\circ$, and $K (y)$ for $y \in V^*$, accordingly. 
\end{description}
We note that, for any $x \in V^*$,  
if $J(x)$ (resp.~$I(x)$) is defined, then it is equal to  
either $\{x\}$ or $R_{H_i}(x)$ (resp.~either $\{x\}$ or $\overline{R_{H_i}(x)}$) 
for some $H_i \in \Lambda$. 
In particular, the last element of $J(x)$ and the first element of $I(x)$ are $x$. 
We also note that $J(x)$ and $I(x)$ are not used in the procedure explicitly, 
but we introduce them to show the validity of the procedure.

\subsection{Grafting a Blossom}
\label{sec:graft}

In this subsection, we describe $\DBlossom$ that replaces a blossom with another blossom, 
which is called in Step (4-2) of $\Search$. 
See Fig.~\ref{fig:revisionfig15} for an illustration.

\begin{description}
\item[Procedure] $\DBlossom(v,H_i)$ 

\item[Step 1:]
Set $H := H_i \cup \{b_i\}$, where $H_i$ is a normal blossom. 
Introduce new vertices $b$ and $t$ in the same say as Step 2 of $\Blossom(v, u)$ 
with $r:=v$ and $g:=b_i$, add $t$ to $H$, and 
apply the pivoting operation around $\{b,t\}$ to $C^*$. 
Label $t$ with $P(t) := P(v)b t$ and $J (t) := \{ t\}$, 
and extend the ordering $\prec$ of the labeled vertices
so that $t$ is just after $v$, i.e., 
$v \prec t$ and no element is between $v$ and $t$.
Set $\rho(b) := v$ and $I(b) := \{b\}$. 

\item[Step 2:]
For each vertex $x\in H^\bullet_i$ in the order of $<_{H_i}$, 
label $x$ with $P(x) := P(v) b t b_i R_{H_i}(x)$ and 
$J(x) := t b_i R_{H_i}(x)$, and put $x$ into the queue.

\item[Step 3:]
Label $H$ with $\oplus$. 
Define $R_{H} (x) := P(x|b)$ for each $x \in H^\bullet$. 
Define $<_{H}$ by the ordering $\prec$ of the labeled vertices in $H^\bullet$.  
Add $H$ to $\Lambda$ with $q(H)=0$ regarding $b$ and $t$ as the bud and the tip of $H$, 
respectively, 
and update $\Lambda_{\rm n}$, $\lambda$, $G^\circ$, and 
$K (y)$ for $y \in V^*$, accordingly.

\item[Step 4:]
Set $\epsilon := q(H_i)$ and modify the dual variables by  
$q(H_i) := 0$, $q(H) := \epsilon$, 
\begin{align*}
p(b_i) &:= 
\begin{cases}
p(b_i) - \epsilon & \mbox{if $b_i \in V^* \setminus B^*$}, \\
p(b_i) + \epsilon & \mbox{if $b_i \in B^*$}, 
\end{cases} \\
p(t) &:= 
\begin{cases}
p(t) - \epsilon & \mbox{if $t \in B^*$}, \\
p(t) + \epsilon & \mbox{if $t \in V^* \setminus B^*$}.
\end{cases}
\end{align*}
Apply $\Expand(H_i)$ to delete $H_i$ from $\Lambda$, and set $H:=H\setminus\{b_i,t_i\}$.  
For each vertex $x$, delete $b_i$ and $t_i$ from  $P(x)$, $R_H(x)$, and $J(x)$. 
\end{description}

\begin{figure}[htbp]
  \centering
    \includegraphics[width=13cm]{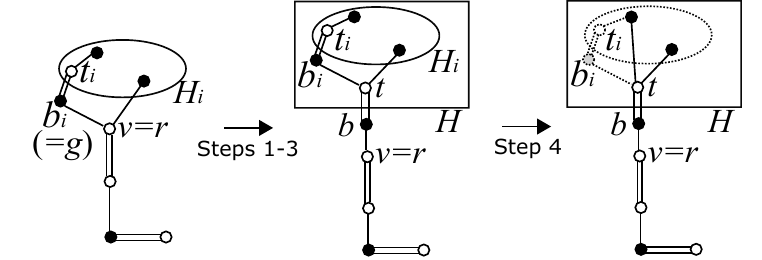}
      \caption{Illustration of $\DBlossom(v,H_i)$.} 
    \label{fig:revisionfig15}
\end{figure}

We note that Step 4 of $\DBlossom(v,H_i)$ is executed to keep the condition $H_i \cap V \not = H_j \cap V$ for distinct $H_i, H_j \in \Lambda$.

\subsection{Basic Properties}
\label{sec:basicpro}

For better understanding of the pivoting operations in $\Blossom(v,u)$ and $\DBlossom(v,H_i)$, 
we give several lemmas in this subsection. 
Then, we show that 
the conditions (BT1), (BT2), and (DF1)--(DF3) hold in the search procedure. 

\begin{lemma}
\label{lem:72}
Suppose that $\Blossom(v, u)$ or {\em Steps 1--3} of $\DBlossom(v,H_i)$ have created a new blossom $H$ containing no source line. 
Then the following conditions hold after the pivoting operation: 
\begin{itemize}
\item
$b$ and $t$ satisfy the conditions in {\em (BT2)}, 
\item
there is no edge in $F^*$ between $r$ and 
$H$, and
\item
there is no edge in $F^*$ between $g$ and 
$V^* \setminus H$. 
\end{itemize}
\end{lemma}

\begin{proof}
To show the properties,  
we use the notation $\hat V^*, \hat B^*$, $\hat C^*$, and $\hat F^*$ to represent the objects 
after the pivoting operation, whereas $V^*$, $B^*$, $C^*$, and $F^*$ represent those before 
the pivoting operation. We only consider the case when 
$b \in \hat V^* \setminus \hat B^*$ and $t \in \hat B^*$ 
as the other case can be dealt with in a similar way. 

In Step 2 of $\Blossom(v, u)$ (or Step 1 of $\DBlossom(v,H_i)$), 
we have $C^*_{bt}=C^*_{rg}\neq 0$, and hence 
$\hat C^*_{t b} = 1/C^*_{bt} \neq 0$.  
Since $C^*_{b y} = 0$ for any $y \in (V^* \setminus B^*) \setminus H$, we have 
$\hat C^*_{t y} = C^*_{b y} / C^*_{b t} = 0$ for any $y \in (\hat V^* \setminus \hat B^*) \setminus H$ with $y\neq b$.  
Similarly, since $C^*_{x t} = 0$ for any $x \in H \cap B^*$, we have  
$\hat C^*_{x b} = C^*_{x t} / C^*_{b t} = 0$ for any $x \in  H \cap \hat B^*$ with $x\neq t$. 
Thus, $b$ and $t$ satisfy the conditions in (BT2). 

Since $C^*_{bt}=C^*_{rt}$ and $C^*_{ry} = C^*_{by}$ for any $y \in (H \setminus B^*) \setminus \{t\}$, 
we have $\hat C^*_{ry} = C^*_{ry} - C^*_{rt} (C^*_{bt})^{-1} C^*_{by} =0$
for any $y \in (H \setminus B^*) \setminus \{t\}$ by Lemma~\ref{lem:pivot}. 
Thus, there is no edge in $\hat F^*$ between $r$ and $H$. 
Similarly, since $C^*_{bt}=C^*_{bg}$ and $C^*_{xg} = C^*_{xt}$ for any $x \in (B^* \setminus H) \setminus \{b\}$, 
we have $\hat C^*_{xg} = C^*_{xg} - C^*_{xt} (C^*_{bt})^{-1} C^*_{bg} =0$ 
for any $x \in (B^* \setminus H) \setminus \{b\}$ by Lemma~\ref{lem:pivot}. 
Thus, there is no edge in $\hat F^*$ between $g$ and $V^* \setminus H$. 
\end{proof}

The following lemma shows 
how creating a blossom affects the edges in $F^\circ$.  

\begin{lemma}\label{clm:71}
Suppose that $\Blossom(v, u)$ or {\em Steps 1--3} of $\DBlossom(v,H_i)$ have created a new blossom $H$ containing no source line, and 
let $F^\circ$ (resp.~$\hat F^\circ$) be the tight edge set before (resp.~after) 
the execution of $\Blossom(v, u)$ or {\em Steps 1--3} of $\DBlossom(v,H_i)$. 
If $(x,y) \in F^\circ \triangle \hat F^\circ$, then  
{\em (i)} $\{x, y\} \cap \{ b, t\} \not= \emptyset$, or
{\em (ii)} exactly one of $\{x, y\}$, say $x$, is contained in $H$, and
$(x, r), (g, y) \in F^\circ$.
\end{lemma}

\begin{proof}
Suppose that $\{x, y\} \cap \{ b, t \} = \emptyset$. 
By Lemma~\ref{lem:pivot}, 
we have $(x,y) \in F^\circ \triangle \hat F^\circ$ only when 
$(x, b), (t, y) \in F^*$ or $(y, b), (t, x) \in F^*$ holds before the pivoting operation
in Step 2 of $\Blossom(v,u)$ (or Step 1 of $\DBlossom(v,H_i)$). This shows that exactly one of $\{x, y\}$, say $x$, 
is contained in $H$, and that $(x,r),(g,y)\in F^*$ holds before $\Blossom(v,u)$ (or $\DBlossom(v,H_i)$). 

Suppose that $x \in B^*$. In this case, if $(x, r), (g, y) \in F^*$ holds before $\Blossom(v,u)$ (or $\DBlossom(v,H_i)$)
and $(x,y) \in F^\circ \triangle \hat F^\circ$, then we have 
\begin{align*}
p(y) - p(x) &= Q_{xy}, \\
p(r)  - p(g) &= Q_{rg}, \\
p(r) - p(x) &\ge Q_{x r}, \\
p(y) - p(g)&\ge Q_{g y}.  
\end{align*}
Furthermore, we have $Q_{xy}+Q_{rg}  =  Q_{xr}+Q_{gy}$ by a simple counting argument. 
Combining these inequalities, we see that all the inequalities above must be tight. 
Thus, we have 
$(x, r), (g, y) \in F^\circ$. 
The same argument can be applied to the case when $x \in V^* \setminus B^*$. 
\end{proof}

The proof of this lemma implies the following result. 
\begin{corollary}\label{cor:78}
Suppose that $\Blossom(v, u)$ or {\em Steps 1--3} of $\DBlossom(v,H_i)$ have created a new blossom $H$ containing no source line, and 
let $F^*$ (resp.~$\hat F^*$) be the edge set before (resp.~after) 
the execution of $\Blossom(v, u)$ or {\em Steps 1--3} of $\DBlossom(v,H_i)$. 
If $(x,y) \in F^* \triangle \hat F^*$, then 
{\em (i)} $\{x, y\} \cap \{ b, t\} \not= \emptyset$, or
{\em (ii)} exactly one of $\{x, y\}$, say $x$, is contained in $H$, and
$(x, r), (g, y) \in F^*$.
\end{corollary}

The following lemma shows that Step 4 of $\DBlossom(v,H_i)$ roughly replaces edges incident to $t_i$ with ones incident to $t$.

\begin{lemma}
\label{lem:dblossomexpand}
Suppose that $\Expand(H_i)$ is executed for some positive blossom $H_i \in \Lambda_{\rm n}$ in $\DBlossom(v,H_i)$. 
Then, we have the following. 
\begin{itemize}
\item
$\Expand(H_i)$ does not affect the edges in $F^*$ that are not incident to $\{t, b_i, t_i\}$. 
\item
If $(t, x) \in F^*$ after $\Expand(H_i)$, then
$(t, x) \in F^*$ or $(t_i, x) \in F^*$ before $\Expand(H_i)$.
\item
If $(t, x) \in F^\circ$ after $\Expand(H_i)$, then
$(t, x) \in F^\circ$ or $(t_i, x) \in F^\circ$ before $\Expand(H_i)$. 
\item
If $(t_i, x) \in F^\circ$ before $\Expand(H_i)$ with $x \not= b_i$, then 
$(t, x) \in F^\circ$ after $\Expand(H_i)$. 
\end{itemize}
\end{lemma}

\begin{proof}
Since $(b_i, t_i)$ is the only edge in $F^*$ connecting $b_i$ and $H_i$, 
$(b_i, t)$ and $(b_i, t_i)$ are the only edges in $F^*$ incident to $b_i$ just before $\Expand(H_i)$. 
Thus, the first property holds. 
By Lemma~\ref{lem:pivotsing}, 
$(t, x) \in F^*$ after $\Expand(H_i)$ if and only if $C^*[\{t, x, b_i, t_i\}]$ is nonsingular before $\Expand(H_i)$, 
which shows the second property. 
Then, by the dual feasibility, we obtain the third property. 
If $(t_i, x) \in F^\circ$ before $\Expand(H_i)$, then 
$(t, x) \not\in F^*$ before $\Expand(H_i)$ by the dual feasibility, 
and hence $C^*[\{t, x, b_i, t_i\}]$ is nonsingular. 
Thus, $(t, x) \in F^\circ$ after $\Expand(H_i)$. 
\end{proof}

We can also see that 
creating a new blossom does not violate the dual feasibility as follows. 

\begin{lemma}
\label{lem:createblossomdual}
Suppose that the dual variables are feasible before $\Blossom(v,u)$ or 
{\em Steps 1--3} of $\DBlossom(v,H_i)$, which 
create a new blossom $H$. Then, the dual variables remain feasible 
after $\Blossom(v, u)$ or {\em Steps 1--3} of $\DBlossom(v,H_i)$. 
\end{lemma}

\begin{proof}
We use $\hat V^*$, $\hat B^*$, $\hat C^*$, $\hat F^*$, $\hat p$, and $\hat \Lambda$ to represent the objects 
after $\Blossom(v,u)$ (or Steps 1--3 of $\DBlossom(v,H_i)$), whereas $V^*$, $B^*$, $C^*$, $F^*$, $p$, and $\Lambda$ represent 
the objects before $\Blossom(v,u)$ (or Steps 1--3 of $\DBlossom(v,H_i)$). 
We only consider the case when 
$b \in \hat V^* \setminus \hat B^*$ and $t \in \hat B^*$, 
as the other case
can be dealt with in a similar way. 

Since there is an edge in $F^\circ$ between $r$ and $g$, 
we have $p(g) - p(r) = Q_{r g}$, and hence
\begin{equation}
\hat p(b) = \hat p(t) = p(r) + Q_{r b} = p(g) + Q_{r b} - Q_{r g} = p(g) - Q_{g b}.
\label{eq:70}
\end{equation}
By the definition of $\hat C^*$, we have the following. 
\begin{itemize}
\item
If $(x, b) \in \hat F^*$ for $x \in B^*$, then 
$x \in V^* \setminus H$ and $(x, g) \in F^*$. 
Thus, we have 
$$
\hat p (b) - \hat p (x) = p(g) - p(x) - Q_{g b} \ge Q_{x g} - Q_{g b} = Q_{x b} 
$$
by (\ref{eq:70}) and the dual feasibility before $\Blossom(v, u)$ (or Steps 1--3 of $\DBlossom(v,H_i)$). 
\item
If $(t, y) \in \hat F^*$ for $y \in V^* \setminus B^*$, then 
$y \in H$ and $(r, y) \in F^*$. 
Thus, we have 
$$
\hat p (y) - \hat p (t) = p(y) - p(r) - Q_{r b} \ge Q_{ry} - Q_{rb} = Q_{b y} = Q_{t y}
$$
by (\ref{eq:70}), the dual feasibility before $\Blossom(v, u)$ (or Steps 1--3 of $\DBlossom(v,H_i)$), and $q(H)=0$. 
\item 
If $(x,y)\in \hat F^*\setminus F^*$ for $x\in B^*$ and $y\in V^*\setminus B^*$,
then $x\in V^*\setminus H$, $y\in H$, and $(x,g), (r,y)\in F^*$ by Corollary~\ref{cor:78}. Thus, we have  
\begin{align*}
\hat p(y) - \hat p(x) = p(y) - p(x)  
&= (p(y)-p(r)) - (p(g)-p(r)) + (p(g)-p(x)) \\
&\geq Q_{r y} - Q_{r g} + Q_{x g}  = Q_{x y}
\end{align*}
by the dual feasibility before $\Blossom(v,u)$ (or Steps 1--3 of $\DBlossom(v,H_i)$). 
\end{itemize}
These facts show that 
$\hat p$ and $\hat q$ are feasible with respect to $\hat\Lambda$. 
\end{proof}

It is obvious that creating a new blossom does not violate (BT1).  
Thus, by Lemmas~\ref{lem:expand}, \ref{lem:72}, and \ref{lem:createblossomdual}, 
we see that the procedure \Search\ keeps the conditions (BT1), (BT2), and (DF1)--(DF3).

\section{Validity}
\label{sec:validity}
This section is devoted to the validity proof of the procedures described in Section~\ref{sec:search}. 
In Section~\ref{sec:propertyR}, 
we introduce properties (BR4) and (BR5) of the routing in blossoms.  
The procedures 
are designed so that they keep the  
conditions (BR1)--(BR5). 
Assuming these conditions, we show in Section~\ref{sec:validpath} that a nonempty output of $\Search$ is 
indeed an augmenting path. 
In Section~\ref{sec:validrouting},  
we show that these conditions hold during the procedure.

\subsection{Properties of Routings in Blossoms}
\label{sec:propertyR}

In this subsection, we introduce properties (BR4) and (BR5) of $R_{H_i}(x)$ kept in the procedure. 
Recall that, for $H_i \in \Lambda$,
\begin{align*}
H^-_i  &= \{ v \in H_i \setminus \{t_i\} \mid \mbox{there is an edge in $F^*$ between $v$ and $V^* \setminus H_i$} \},  \\
H^\bullet_i &= H^-_i \cup (H_i \cap V).
\end{align*}
In addition to (BR1)--(BR3), we assume that $R_{H_i}(x)$ satisfies the following 
(BR4) and (BR5) for any $H_i \in \Lambda$ and $x \in H^\bullet_i$.

\begin{description}
\item[(BR4)]
$G^\circ [R_{H_i}(x) \setminus \{x\}]$ 
 has a unique tight perfect matching.  
\item[(BR5)] 
If $x \in H^-_i$, then we have the following. 
Suppose that $Z \subseteq R_{H_i}(x)  \cap H^-_i$ satisfies that 
$z \ge_{H_i} x$ for any $z \in Z$, $Z \not= \{x\}$, and 
$|H_j \cap Z| \le 1$ for any positive blossom $H_j \in \Lambda$. 
Then, $G^\circ [R_{H_i}(x) \setminus Z]$ has no tight perfect matching. 
\end{description}
Here, we suppose that $G^\circ [\emptyset]$ 
 has a unique tight perfect matching $M= \emptyset$ to simplify the description.  

We now explain roles of (BR4) and (BR5). 
These conditions are used to show that 
the output $P$ in Step (3-1) of $\Search$ satisfies (AP3), i.e., 
$G^\circ [P]$ has a unique tight perfect matching. 
We will show that the obtained path $P$ can be decomposed into subsequences, 
and each subsequence consists of a singleton or a set $R_{H_i}(x)$ for some $x \in H^\bullet_i$
(see Lemma~\ref{lem:dec}). 
Our aim is to show that if $G^\circ [P]$ has a tight perfect matching, then 
$x$ is the only vertex in $R_{H_i}(x)$ 
that is matched with a vertex outside $R_{H_i}(x)$. 
This is guaranteed by (BR5), where $Z$ means the set of vertices that are matched with vertices outside $R_{H_i}(x)$.
Then, (BR4) assures that there exists a unique perfect matching covering $R_{H_i}(x)$ except $x$.

\subsection{Finding an Augmenting Path}
\label{sec:validpath}

This subsection is devoted to the validity of Step (3-1) of $\Search$. 
We first show the following lemma.

\begin{lemma}
\label{lem:dec}
In each step of $\Search$,  
for any labeled vertex $x$, $P(x)$ is decomposed as 
$$P(x)=J(x_k) I(y_k) \cdots J(x_1)I(y_1)J(x_0)$$
with $x_k \prec \cdots \prec x_1 \prec x_0 = x$ such that, 
for each $i$, 
\begin{description}
\item[(PD0)]
$J(x_i)$ is equal to either 
$\{x_i\}$ or $R_{H_j}(x_i)$  
for some $H_j \in \Lambda$, and
$I(y_i)$ is equal to either 
$\{y_i\}$ or $\overline{R_{H_j}(y_i)}$ 
for some positive blossom $H_j \in \Lambda$,

\item[(PD1)]
$x_i$ is adjacent to $y_i$ in $G^\circ$, 
\item[(PD2)]
the first element of $J(x_{i-1})$ and the last element of $I(y_i)$ form a line or a dummy line, 
\item[(PD3)]
any labeled vertex $z$ with $z \prec x_i$ is not adjacent to $I(y_i) \cup J(x_{i-1})$ in $G^\circ$, and 
\item[(PD4)]
$x_i$ is not adjacent to $J(x_{i-1})$ in $G^\circ$. 
Furthermore, if $I(y_i) = \overline{R_{H_j}(y_i)}$, then
$x_i$ is not adjacent to $\{z \in I(y_i) \mid z <_{H_j} y_i \}$ in $G^\circ$. 
\end{description}
See Fig.~\ref{fig:revisionfig16} for an example of the decomposition. 
\end{lemma}

\begin{figure}[htbp]
  \centering
    \includegraphics[width=14cm]{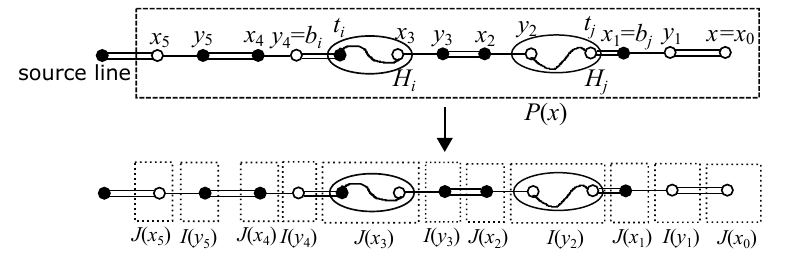}
    \caption{An example of the decomposition.} 
    \label{fig:revisionfig16} 
\end{figure}

\begin{proof}
The procedure $\Search$ naturally defines the decomposition
$$P(x)=J(x_k) I(y_k) \cdots J(x_1)I(y_1)J(x_0).$$  
It suffices to show that $\Blossom(v, u)$ and $\DBlossom(v, H_i)$ do not violate the conditions (PD0)--(PD4), 
since we can easily see that the other operations do not violate them. 

We first consider the case when $\Blossom(v, u)$ is applied to obtain a new blossom $H$. 
In $\Blossom(v, u)$, $P(x)$ is updated or defined as 
$P(x) := P(x)$, $P(x) := P(r) bt P(x|r)$, or $P(x) := P(r) b R_{H}(x)$. 
Let $F^\circ$ (resp.~$\hat F^\circ$) be the tight edge sets before (resp.~after) 
the execution of $\Blossom(v, u)$ that adds $H$ to $\Lambda$.

Suppose that $P(x)$ is defined by $P(x) := P(r) I(b) J(x)$, 
where $I(b) = \{b\}$ and $J(x) = R_{H}(x)$. 
In this case, (PD0), (PD1), and (PD2) are trivial. 
We now consider (PD3). 
Since $P(r)$ satisfies (PD3), 
in order to show that 
any labeled vertex $z$ with $z \prec x_i$ is not adjacent to $I(y_i) \cup J(x_{i-1})$ in $\hat G^\circ=(V^*,\hat{F}^\circ)$, 
it suffices to consider the case when $x_i = r$, $y_i = b$, and $x_{i-1} = x$ (see Fig.~\ref{fig:revisionfig18}). 
Assume to the contrary that $z \prec r$ is adjacent to $I(b) \cup J(x)$ in $\hat G^\circ$. 
Since $z$ is not adjacent to $I(b) \cup J(x)$ in $G^\circ$ by the procedure, 
Lemma~\ref{clm:71} shows that $(z, g) \in F^\circ$. 
This contradicts that $z \prec x_i =r$ and the definition of $H$. 
To show (PD4), 
it suffices to consider the case when $x_i = r$. 
In this case, since $r$ is not adjacent to $H$ in $\hat G^\circ$ by 
Lemma~\ref{lem:72}, 
$P(x)$ satisfies (PD4).

\begin{figure}[htbp]
  \centering
    \includegraphics[width=8cm]{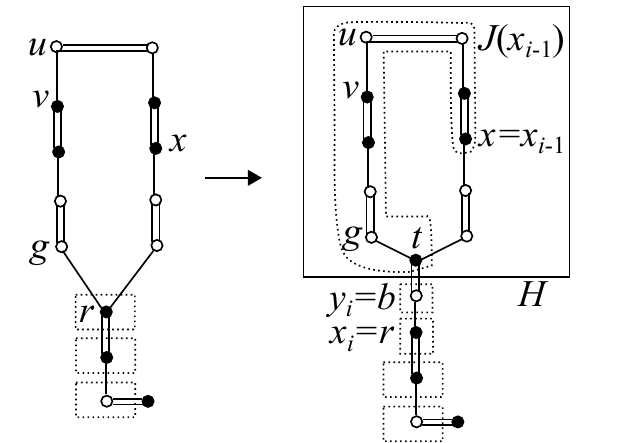}
    \caption{The case of $P(x) := P(r) I(b) J(x)$.} 
    \label{fig:revisionfig18} 
\end{figure}

Suppose that 
$P(x)$ is updated as
$P(x) := P(x)$ or
$P(x) := P(r) I(b) J(t) P(x|r)$, 
where $I(b) = \{b\}$ and $J(t) = \{t\}$ (see Fig.~\ref{fig:revisionfig19} for an example). 
In this case, (PD0), (PD1), and (PD2) are trivial. 
We now consider (PD3). 
Since (PD3) holds before creating $H$, 
in order to show that 
any labeled vertex $z$ with $z \prec x_i$ is not adjacent to $w \in I(y_i) \cup J(x_{i-1})$ in $\hat G^\circ$, 
it suffices to consider the case when 
(i) $z = t$, or
(ii) $w \in I(b) \cup J(t)$, or
(iii) $(z, g) \in F^\circ$ and $(w, r) \in F^\circ$, or 
(iv) $(w, g) \in F^\circ$ and $(z, g) \in F^\circ$ by Lemma~\ref{clm:71}. 
In the first case, if $(t, w) \in \hat F^\circ$, then $(r, w) \in F^\circ$, 
which contradicts that (PD3) holds before creating $H$. 
In the second case, if $w=b$, then $(z, w) \in \hat F^\circ$ implies that $(z, g) \in F^\circ$,
which contradicts that $z \prec x_i = r$ and the definition of $H$. 
If $w=t$, then $(w, z) \in \hat F^\circ$ implies that $(r, z) \in F^\circ$, 
which contradicts that $r$ and $z$ are labeled.
In the third case, $(w, r) \in F^\circ$ implies $x_i \preceq r$ as (PD3) holds before creating $H$. 
By the definition of $H$, however, $z \prec x_i \preceq r$ contradicts $(z, g) \in F^\circ$. 
In the fourth case, $(z, r) \in F^\circ$ contradicts that $r$ and $z$ are labeled. 
By these four cases, we obtain (PD3).

\begin{figure}[htbp]
  \centering
    \includegraphics[width=8cm]{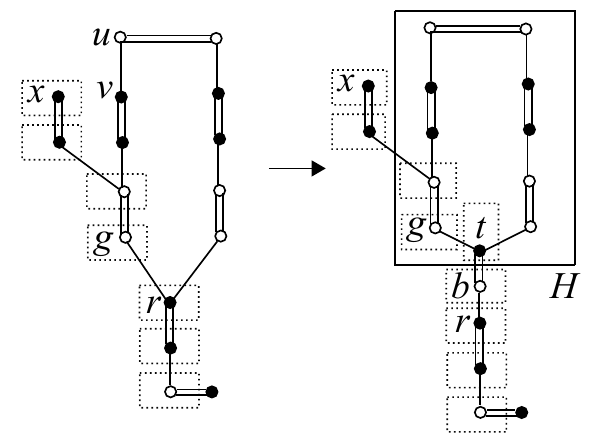}
    \caption{The case of $P(x) := P(r) I(b) J(t) P(x|r)$.} 
    \label{fig:revisionfig19} 
\end{figure}

We next consider (PD4). 
Since (PD4) holds before creating $H$, 
in order to show that $x_i$ is not adjacent to $w \in J(x_{i-1})$
or $w \in \{z \in I(y_i) \mid  z <_{H_j} y_i \}$ in $\hat F^\circ$
it suffices to consider the case when 
(i) $x_i = r$, or 
(ii) $x_i = t$, or
(iii) $(x_i, w)$ crosses $H$.  
In the first case, the claim is obvious. 
In the second case, if $(t, w) \in \hat F^\circ$, then $(r, w) \in F^\circ$, 
which contradicts that (PD4) holds before creating $H$. 
In the third case, since $x_i \in H$ and $w \not\in H$, Lemma~\ref{clm:71} implies that 
it suffices to consider the case when $(w, g) \in F^\circ$ and $(x_i, r) \in F^\circ$,  
which contradicts that $x_i$ and $r$ are labeled.
By these three cases, we obtain (PD4).

We can show that $\DBlossom(v, H_i)$ does not violate (PD0)--(PD4) in a similar manner
by observing that $P(x)$ is updated or defined as 
$P(x) := P(x)$ or $P(x) := P(v) R_H(x)$ in $\DBlossom(v, H_i)$. 
We note that 
$\Expand(H_i)$ in $\DBlossom$ does not affect (PD0)--(PD4)
by Lemma~\ref{lem:dblossomexpand}. 
\end{proof}

Recall that we assume the conditions (BT1), (BT2), (DF1)--(DF3), and (BR1)--(BR5).  
We are now ready to show the validity of Step (3-1) of $\Search$. 

\begin{lemma}
\label{lem:aug}
If $\Search$ returns $P:= P(v) \overline{P(u)}$ in {\em Step (3-1)}, 
then $P$ is an augmenting path.  
\end{lemma}
\begin{proof}
It suffices to show that
$G^\circ[P]$ has a unique tight perfect matching. 
By Lemma~\ref{lem:dec}, $P(v)$ and $P(u)$ 
are decomposed as $P(v)=J(v_k)I(s_k)\cdots J(v_1)I(s_1)J(v_0)$ and 
$P(u)=J(u_l)I(r_l)\cdots J(u_1)I(r_1)J(u_0)$. For each pair of $i\leq k$ and $j\leq l$, 
let $X_{ij}$ denote the set of vertices in the subsequence 
$$
J(v_i)I(s_i)\cdots J(v_1)I(s_1)J(v_0)\overline{J(u_0)} \, \overline{I(r_1)}\,
\overline{J(u_1)}\cdots \overline{I(r_j)}\,\overline{J(u_j)}
$$ of $P$. 
We intend to show inductively that $G^\circ[X_{ij}]$ has 
a unique tight perfect matching. 

We first show that $G^\circ[X_{00}] = G^\circ[J(u)\cup J(v)]$ has a unique tight perfect matching. 
Since $u$ and $v$ are adjacent in $G^\circ$, (PD0) and (BR4) guarantee that $G^\circ[J(u)\cup J(v)]$ has a tight perfect matching.
Let $M$ be an arbitrary tight perfect matching in $G^\circ[J(u)\cup J(v)]$,
and let $Z$ be the set of vertices in $J(v)$ adjacent to $J(u)$ in $M$. 
If $J(v) = \{v\}$, then it is obvious that $Z = \{v\}$. 
Otherwise, $J(v) = R_{H_i}(v)$ for some $H_i \in \Lambda$.   
For any positive blossom $H_j \in \Lambda$, since $M$ is consistent with $H_j$ by the definition of a tight matching, 
we have that $|H_j \cap Z| \le 1$. 
Since there are no edges of $G^\circ$ between $J(u)$ and $\{y \in J(v) \mid y \prec v \}$, 
we have that $z \ge_{H_i} v$ for any $z \in Z$. 
Furthermore, since there is an edge in $M$ connecting each $z \in Z$ and $J(u)$, we have $Z \subseteq J(v) \cap H^-_i$. 
Then it follows from (BR5) that 
$G^\circ[J(v)\setminus Z]$ has no tight perfect matching unless $Z=\{v\}$. 
This means $v$ is the only vertex in $J(v)$ adjacent to $J(u)$ in $M$.
Note that $G^\circ[J(v)\setminus\{v\}]$ has a unique tight perfect matching by (BR4), 
which must form a part of $M$. Let $z$ be the vertex adjacent to $v$ in $M$. 
Since the vertices in $\{y \in J(u) \mid y \prec u \}$ are not adjacent to $v$ in $G^\circ$, 
we have $z \ge_{H_j} u$ if $J(u) = R_{H_j}(u)$ for some $H_j \in \Lambda$ (see Fig.~\ref{fig:revisionfig20}). 
By (BR5) again,  $G^\circ[J(u)\setminus\{z\}]$ has 
no tight perfect matching unless $z=u$. This means $M$ must contain the edge $(u, v)$. 
Note that $G^\circ[J(u)\setminus\{u\}]$ has a unique tight perfect matching by (BR4), 
which must form a part of $M$. Thus $M$ must be the unique tight perfect matching in 
$G^\circ[J(u)\cup J(v)]$. 

\begin{figure}[htbp]
  \centering
    \includegraphics[width=11cm]{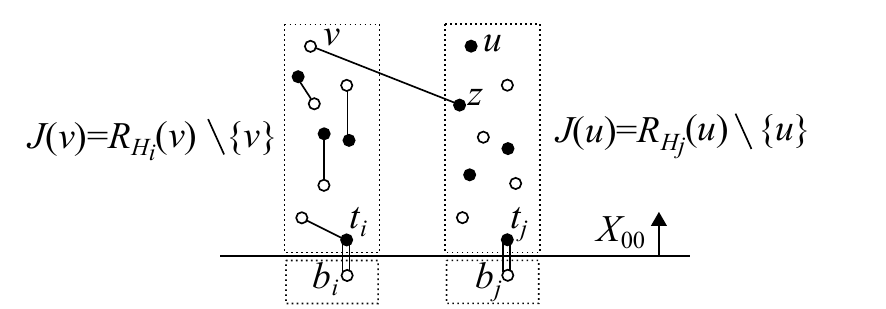}
    \caption{An example of $G^\circ[X_{00}]$. Real lines represent the edges in $M$.} 
    \label{fig:revisionfig20} 
\end{figure}

We now show the statement for general $i$ and $j$ assuming that the same 
statement holds if either $i$ or $j$ is smaller. 
Suppose that $v_i \prec u_j$. Then there are no edges of $G^\circ$ 
between $X_{ij}\setminus J(v_i)$ and $\{y \in J(v_i) \mid y \prec v_i \}$ by (PD3) of Lemma~\ref{lem:dec}.  
Let $M$ be an arbitrary tight perfect matching in $G^\circ[X_{ij}]$,  
and let $Z$ 
be the set of vertices in $J(v_i)$ adjacent to $X_{ij}\setminus J(v_i)$ in $M$. 
Then, by the same argument as above, 
$G^\circ[J(v_i)\setminus Z]$ has no tight perfect matching unless $Z=\{v_i\}$. 
Thus $v_i$ is the only vertex in $J(v_i)$ matched to $X_{ij}\setminus J(v_i)$ in $M$.  
Since $v_i$ is not adjacent to $X_{i-1,j}$ in $G^\circ$ by (PD3) and (PD4) of Lemma~\ref{lem:dec}, 
an edge connecting $v_i$ and $I(s_i)$ must belong to $M$. 
We note that it is the only edge in $M$ between $I(s_i)$ and $X_{ij} \setminus I(s_i)$ since $M$ is tight 
and $I(s_i)$ is equal to either 
$\{s_i\}$ or $\overline{R_{H}(s_i)}$ 
for some positive blossom $H \in \Lambda$. 
Let $z$ be the vertex adjacent to $v_i$ in $M$. 
By (BR5), $G^\circ[I(s_i)\setminus \{z\}]$ has no tight perfect matching unless $z = s_i$ (see Fig.~\ref{fig:revisionfig21}). 
This means that $M$ contains the edge $(v_i, s_i)$. 
Note that each of $G^\circ[J(v_i)\setminus\{v_i\}]$ and $G^\circ[I(s_i)\setminus \{s_i\}]$ has 
a unique tight perfect matching by (BR4), and so does $G^\circ[X_{i-1,j}]$ by induction hypothesis. 
Therefore, $M$ is the unique tight perfect matching in $G^\circ[X_{ij}]$. 
The case of $v_i \succ u_j$ can be dealt with similarly. 
Thus, we have seen that $G^\circ[X_{kl}] = G^\circ[P]$ has a unique tight perfect matching.
\end{proof}

\begin{figure}[htbp]
  \centering
    \includegraphics[width=14cm]{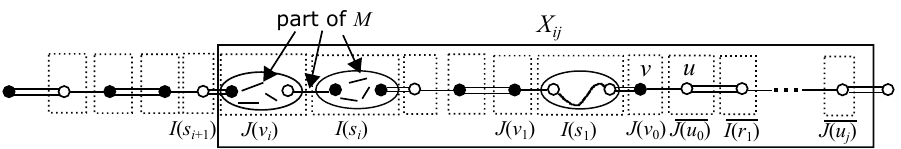}
    \caption{An example of $G^\circ[X_{ij}]$.} 
    \label{fig:revisionfig21} 
\end{figure}

This proof implies the following corollaries. 

\begin{corollary}
\label{cor:reachable}
For any labeled vertex $v \in V^*$, 
$G^\circ[P(v) \setminus \{v\}]$ has a unique tight perfect matching. 
\end{corollary}

\begin{corollary}
\label{cor:uniqueedgeaug}
If $\Search$ returns $P$, 
then the unique tight matching in $G^\circ[P]$ 
contains exactly one edge connecting $H_i$ and $V^* \setminus H_i$
for each $H_i \in \Lambda$ with $P \cap H_i \not= \emptyset$. 
\end{corollary}

\subsection{Routing in Blossoms}
\label{sec:validrouting}

First, to see that $R_H(x)$ is well-defined for each $x \in H^\bullet$ when we create a new blossom $H$, 
we observe that every vertex $x \in H^\bullet$ satisfies one of the six cases in Step 4 of $\Blossom(v, u)$. 
This is because, if $x \in H_i \setminus H^\bullet_i$ for some $H_i \in \Lambda$ with $H_i \subsetneq H$, then $x \not\in H^\bullet$, and   
if $c \not=d$, $K(c) = H_i \cup \{b_i\}$, and $x = b_i = g$ for some $H_i \in \Lambda_{\rm n}$, then $x \not\in H^-$ by Lemma~\ref{lem:72}.

When we create a new blossom $H$ in $\DBlossom(v, H_i)$, 
for each $x \in H^\bullet$, 
$R_{H}(x)$ clearly satisfies (BR1)--(BR5) by Lemma~\ref{lem:dblossomexpand}.
Suppose that a new blossom $H$ is created in $\Blossom(v, u)$.
For each $x \in H^\bullet$, 
$R_{H}(x)$ defined in $\Blossom(v, u)$ also satisfies 
(BR1)--(BR3). 
We will show (BR4) and (BR5) in this subsection.

\begin{lemma}
\label{lem:BR45}
Suppose that $\Blossom(v, u)$ creates a new blossom $H$. 
Then, for each $x \in H^\bullet$, 
$R_{H}(x)$ satisfies {\em (BR4)} and {\em (BR5)}. 
\end{lemma}

\begin{proof}
We only consider the case when $H$ contains no source line, 
since the case with a source line can be dealt with in a similar way. 
We note that a vertex $y \in H$ is adjacent to $r$ in $G^\circ$ before $\Blossom(v,u)$
if and only if $y$ is adjacent to $t$ in $G^\circ$ after $\Blossom(v,u)$. 
If $x = t$, the claim is obvious. 
We consider the other cases separately. 

\medskip
\noindent
\textbf{Case 1.}
Suppose that $x \in H^\bullet$ was not labeled before $H$ is created. 

Among six cases in Step 4 of $\Blossom(v,u)$, we consider the cases of  (i), (iii), and (v), since
the other cases can be dealt with in a similar manner. 

By Lemma~\ref{lem:dec}, 
$P(v)$ can be decomposed as 
$$P(v)=P(r) b t I(s_k) J(v_{k-1}) I(s_{k-1}) \cdots J(v_1) I(s_1) J(v_0)$$ 
with $v=v_0$. 
In the cases of (i) and (iii), 
$P(u|x)$ can be decomposed 
as $J(u_l) I(r_l) \cdots J(u_1)I(r_1)J(u_0)$ with $u_0=u$, where the first element of 
$J(u_l)$ is $\bar x$, and hence  
$$R_{H}(x)= J(v_k) I(s_k) J(v_{k-1}) \cdots I(s_1)J(v_0)
\overline{J(u_0)}\, \overline{I(r_1)} \cdots \overline{I(r_l)}\,\overline{J(u_l)} x$$ with $v_k=t$.
Similarly, in the case of (v), 
$R_{H}(x)$ can be decomposed as
$$R_{H}(x)= J(v_k) I(s_k) J(v_{k-1}) \cdots I(s_1)J(v_0)
\overline{J(u_0)}\, \overline{I(r_1)} \cdots \overline{I(r_l)}\,\overline{J(u_l)} R_{H_i}(x).$$ 
Therefore, in the cases of (i), (iii), and (v),
we have 
$$R_{H}(x)= J(v_k) I(s_k) J(v_{k-1}) \cdots I(s_1)J(v_0)
\overline{J(u_0)}\, \overline{I(r_1)} \cdots \overline{I(r_l)}\,\overline{J(u_l)}\, \overline{I(r_{l+1})}$$ with $v_k=t$ and $r_{l+1} = x$
(see Fig.~\ref{fig:72} for an example).

\begin{figure}[htbp]
  \centering
    \includegraphics[width=5.5cm]{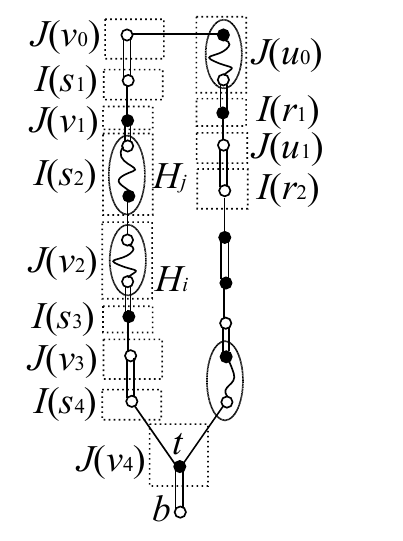}
    \caption{A decomposition of $R_{H}(x)$. In this example, $J(v_1) = \{b_j\}$, 
               $I(s_2) = R_{H_j}(s_2)$, $J(v_2) = R_{H_i}(v_2)$, $I(s_3) = \{b_i\}$, and $J(v_4)=\{ t\}$.} 
    \label{fig:72} 
\end{figure}

We now intend to show that $R_{H}(x)$ satisfies (BR5), that is, $G^\circ[R_{H}(x) \setminus Z]$ has no tight perfect matching
if $Z \subseteq R_{H}(x) \cap H^-$ satisfies that 
$z \ge_{H} x$ for any $z \in Z$, $Z \not= \{x\}$, and $|H_j \cap Z| \le 1$ for any positive blossom $H_j \in \Lambda$. 
Suppose to the contrary that $G^\circ[R_{H}(x) \setminus Z]$ has a tight perfect matching $M$. 
Note that $Z \subseteq I(r_{l+1}) \cup \bigcup_i I(s_i)$, because $z \ge_{H} x$ for any $z \in Z$. 
For each $i$, since either $I(s_i) = \{s_i\}$ or $I(s_i) = R_{H_j}(s_i)$ for some positive blossom $H_j \in \Lambda$, 
we have $|I(s_i) \cap Z| \le 1$. Similarly, $|I(r_{l+1}) \cap Z| \le 1$.
Furthermore, 
if $|I(s_i) \cap Z| = 1$ (resp.~$|I(r_{l+1}) \cap Z| = 1$), then $|I(s_i) \setminus Z|$ (resp.~$|I(r_{l+1}) \setminus Z|$) is even, 
and hence there is no edge in $M$ between $I(s_i)$ (resp.~$I(r_{l+1})$) and its outside, 
because $M$ is a tight perfect matching.
If $Z \subseteq I(r_{l+1})$, then $|I(r_{l+1}) \cap Z| = 1$ and $M$ contains no edge between $I(r_{l+1})$ and the outside of $I(r_{l+1})$, 
which contradicts that $G^\circ[I(r_{l+1}) \setminus Z]$ has no tight perfect matching by (BR5). 
Thus, we may assume that $Z \cap \bigcup_i I(s_i) \not= \emptyset$. 
Since $I(s_i) \cap Z \not= \emptyset$ implies that there exists no edge in $M$ between $I(s_i)$ and the outside of $I(s_i)$, 
we can take the largest number $j$ such that $(v_j, s_j) \notin M$. 
We consider the following two cases separately. 

\medskip

\textbf{Case 1a.} Suppose that $j = k$. 
In this case, since $J(v_k) = \{ t\}$, there exists an edge in $M$ between $t$ and $I(r_{l+1}) \cup (I(s_{k}) \setminus \{s_k\})$. 
See Fig.~\ref{fig:revisionfig22} for an example.
If this edge is incident to $z \in I(s_{k}) \setminus \{s_k\}$, then 
$I(s_k) = R_{H'}(s_k)$ for some positive blossom $H' \in \Lambda$ and $z >_{H'} s_k$ by the procedure, 
and hence $G^\circ[I(s_{k}) \setminus \{z\}]$ has no tight perfect matching by (BR5), which is a contradiction. 
Otherwise, since $v_k = t$ is matched with some vertex $y \in I(r_{l+1})$, 
we have 
$h \in I(r_{l+1})$,  where $h$ is as in Step 4 of $\Blossom(v, u)$. 
This shows that 
$Z \subseteq I(r_{l+1}) \cup I(s_{k})$ as $z \ge_{H} x  = r_{l+1}$ for any $z \in Z$. 
Since $|Z \cap I(r_{l+1})| \le 1$, $|Z \cap I(s_{k})| \le 1$, and
$M$ is a tight perfect matching, we have 
$I(r_{l+1}) \cap Z = \emptyset$, 
$Z = \{ z\}$ for some $z \in I(s_k)$, and 
each of $G^\circ[I(r_{l+1}) \setminus \{y\}]$ and $G^\circ[I(s_{k}) \setminus \{z\}]$ has a tight perfect matching.
This shows that 
$y \le_{H} r_{l+1}$ and $z \le_{H} s_k$
 by (BR5) and the definition of $\le_H$.
Then, we obtain 
$$h \le_{H} y \le_{H} r_{l+1} = x \le_{H} z \le_{H} s_k=g.$$
Since no element is chosen between $g$ and $h$ in Step 4 of $\Blossom(v, u)$,
we have $h=y = r_{l+1} = x$ and $z=s_k=g$,
which contradicts that $z \in H^-$ and $g \not \in H^-$ by Lemma~\ref{lem:72}. 

We note that 
when we apply the same argument to the cases of (ii), (iv), and (vi) 
by changing the roles of $g$ and $h$, 
we obtain $g=y = r_{l+1} = x$.  
Then, this contradicts that $x \in H^-$ and $g \not \in H^-$.

\begin{figure}[htbp]
 \begin{minipage}{0.5\hsize}
  \begin{center}
    \includegraphics[width=5cm]{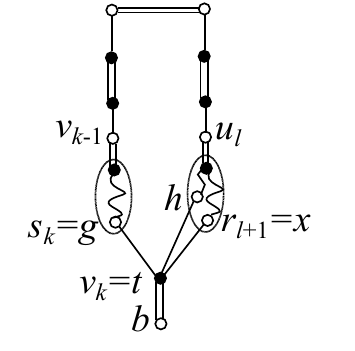}
    \caption{Example of Case 1a.} 
    \label{fig:revisionfig22} 
  \end{center}
 \end{minipage}
 \begin{minipage}{0.5\hsize}
  \begin{center}
    \includegraphics[width=5cm]{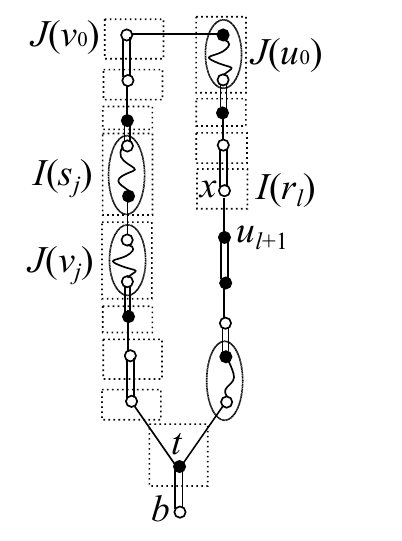}
    \caption{Example of Case 1b.} 
    \label{fig:revisionfig23} 
  \end{center}
 \end{minipage}
\end{figure}

\medskip

\textbf{Case 1b.} Suppose that $j \le k-1$. 
In this case, since $M$ is a tight perfect matching, for $i=j+1, \dots , k$, we have  
$Z \cap I(s_i) = \emptyset$ and $(v_i, s_i)$ is the only edge in $M$ between $I(s_i)$ and the outside of $I(s_i)$. 
We can also see that $Z \cap J(v_j) = \emptyset$, since $z \ge_{H} x$ for any $z \in Z$. 
We denote by $Z_j$ the set of vertices in $J(v_j)$ matched by $M$ to the outside of $J(v_j)$. 
Since $z \ge_{H} x$ for any $z \in Z$ and $Z \cap I(s_i) \not= \emptyset$ for some $i \le j-1$, 
we have $v_j \prec u_{l+1}$, where $u_{l+1}$ is the vertex naturally defined by the decomposition of $P(u)$ (see Fig.~\ref{fig:revisionfig23}).  
Note that the assumption $j \le k-1$ is used here. 
Then, for any vertex $y \in J(v_j)$ with $y <_{H} v_j$, there is no edge in $M$ connecting 
$y$ and $R_{H}(x) \setminus J(v_j)$ because of the following: 
\begin{itemize}
\item
By (PD3) of Lemma~\ref{lem:dec}, $y$ is not adjacent to $I(s_i) \cup J(v_{i-1})$ for $i \le j$, because $y \prec v_j \preceq v_i$. 
\item
By (PD3) of Lemma~\ref{lem:dec}, $y$ is not adjacent to $I(r_i) \cup J(u_{i-1})$ for $i \le l + 1$, because $y \prec v_j \prec u_{l+1} \preceq u_i$. 
\item
If $z \in J(v_i)$ with $i > j$, then $z$ is not adjacent to $y$ by (PD3) of Lemma~\ref{lem:dec}. 
\item
For $i>j$, $(v_i, s_i)$ is the only edge in $M$ between $I(s_i)$ and its outside, and hence 
there is no edge is $M$ between $I(s_i)$ and $y$.
\end{itemize}
This shows that 
$(Z \cap J(v_j)) \cup Z_j = Z_j \subseteq \{ y \in J(v_j) \mid y \ge_{H} v_j\}$. 
Therefore, by (BR5), if $G^\circ [J(v_j) \setminus (Z \cup Z_j)]$ has a tight perfect matching, then 
$Z_j = \{v_j\}$. 
The vertex $v_j$ is not adjacent to the vertices 
in $R_{H}(x)\setminus(J(v_j)\cup I(s_j) \cup \cdots \cup I(s_k))$
by (PD3) and (PD4) of Lemma~\ref{lem:dec}. 
Since $(v_i, s_i)$ is the only edge in $M$ between $I(s_i)$ and its outside for $i>j$,  
$v_j$ has to be adjacent to $I(s_j)$. 
Furthermore, by $(v_j, s_j) \not\in M$ and by (PD4) of Lemma~\ref{lem:dec}, 
we have that 
$v_j$ is incident to a vertex $z \in I(s_j)$ with $z >_{H'} s_j$, 
where $I(s_j) = R_{H'}(s_j)$ for some positive blossom $H' \in \Lambda$. 
Since $G^\circ[I(s_j) \setminus \{z\}]$ has no tight perfect matching by (BR5), we obtain a contradiction. 

\medskip

We next show that $R_{H}(x)$ satisfies (BR4), that is, 
 $G^\circ[R_{H}(x) \setminus \{x\}]$ has a unique tight perfect matching. 
Let $M$ be an arbitrary tight perfect matching in $G^\circ[R_{H}(x) \setminus \{x\}]$. 
Recall that $r_{l+1} = x$ and either $I(r_{l+1}) = \{r_{l+1}\}$ or $I(r_{l+1}) = R_{H_j}(r_{l+1})$ 
for some positive blossom $H_j \in \Lambda$. 
Since $M$ is a tight perfect matching and $|I(r_{l+1}) \setminus \{x\}|$ is even, 
there is no edge in $M$ between $I(r_{l+1})$ and its outside. 
By (BR4), $G^\circ[I(r_{l+1}) \setminus \{x\}]$ has a unique tight perfect matching, 
which must form a part of $M$.  
On the other hand, 
$$G^\circ[J(v_k) I(s_k) J(v_{k-1}) I(s_{k-1})\cdots J(v_1)I(s_1)J(v_0)
\overline{J(u_0)}\, \overline{I(r_1)}\, \overline{J(u_1)}\cdots \overline{I(r_l)}\,\overline{J(u_l)}]$$
has a unique tight perfect matching by the same argument as Lemma~\ref{lem:aug}. 
By combining them, we have that
$G^\circ[R_{H}(x) \setminus \{x\}]$ has a unique tight perfect matching.

\medskip
\noindent
\textbf{Case 2.}
Suppose that $x \in H$ was labeled before $H$ is created. 

We consider the case of $x \in K(y)$ with $y \in P(v|c)$. 
The case of $x \in K(y)$ with $y \in P(u|d)$ can be dealt with in a similar manner. 
By Lemma~\ref{lem:dec}, $R_{H}(x)$ can be decomposed as 
$$R_{H}(x)= J(v_k) I(s_k) J(v_{k-1}) I(s_{k-1})\cdots J(v_{l+1})I(s_{l+1})J(v_l)$$
with $x = v_l$ (see Fig.~\ref{fig:revisionfig24}). 

\begin{figure}[htbp]
  \begin{center}
    \includegraphics[width=3.3cm]{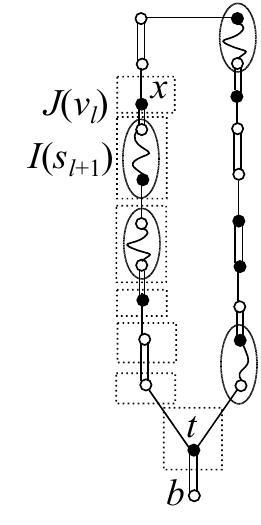}
    \caption{Example of Case 2.} 
    \label{fig:revisionfig24} 
  \end{center}
\end{figure}

We first show that $R_{H}(x)$ satisfies (BR5), that is, 
$G^\circ[R_{H}(x) \setminus Z]$ has no tight perfect matching
if $Z \subseteq R_{H}(x)  \cap H^-$ satisfies that 
$z \ge_{H} x$ for any $z \in Z$, $Z \not= \{x\}$, and $|H_j \cap Z| \le 1$ for any positive blossom $H_j \in \Lambda$. 
Since $z \ge_{H} x$ for any $z \in Z$, 
we have that $Z \subseteq J(v_l) \cup \bigcup_i I(s_i)$, 
which shows that we can apply the same argument as Case 1 to obtain (BR5).  

We next show that $R_{H}(x)$ satisfies (BR4), that is, 
$G^\circ[R_{H}(x) \setminus \{x \}]$ has a unique tight perfect matching. 
By Corollary~\ref{cor:reachable}, $G^\circ[P(x) \setminus \{x\}]$
has a unique tight perfect matching $M$, and a part of $M$ forms a tight perfect matching in 
$G^\circ[R_{H}(x) \setminus \{x \}]$. 
Thus, this matching is a unique tight perfect matching in $G^\circ[R_{H}(x) \setminus \{x \}]$. 
\end{proof}

We note that, for a blossom $H \in \Lambda$,  
creating/deleting another blossom $H'$ does not change $H^-$ and $H^\bullet$
by Corollary~\ref{cor:78} and Lemma~\ref{lem:dblossomexpand}.
We also note that if $R_H(x)$ satisfies (BR1)--(BR5) for $x \in H^\bullet$, 
then creating/deleting another blossom $H'$ does not violate
these conditions by Lemmas~\ref{lem:72},~\ref{clm:71} and~\ref{lem:dblossomexpand}. 
Therefore, Lemma~\ref{lem:BR45} shows that 
the procedure \Search\ keeps the conditions (BR1)--(BR5).

\section{Dual Update}
\label{sec:dualupdatealgo}

In this section, we describe how to modify the dual variables 
when $\Search$ returns $\emptyset$ in Step 2. 
In Section~\ref{sec:infeasible}, 
we show that the procedure keeps the dual variables finite 
as long as the instance has a parity base. 
In Section~\ref{sec:iterations},
we bound the number of dual updates per augmentation. 

Let $R\subseteq V^*$ be the set of vertices that are reached or examined by the search procedure and not contained 
in any blossoms. 
We denote by $R^+$ and $R^-$ the sets of labeled and 
unlabeled vertices in $R$, respectively. In particular, 
the bud $b_i$ of a maximal blossom $H_i$ belongs to 
$R^+$ if $H_i$ is labeled with $\ominus$, and to $R^-$ 
if $H_i$ is labeled with $\oplus$. 
Let $Z$ denote the set 
of vertices in $V^*$ contained in labeled blossoms. The set $Z$ is partitioned into $Z^+$ and $Z^-$, 
where 
\begin{align*}
Z^+ &= \bigcup \{ H_i \mid \mbox{$H_i$ is a maximal blossom labeled with $\oplus$} \}, \\
Z^- &= \bigcup \{ H_i \mid \mbox{$H_i$ is a maximal blossom labeled with $\ominus$} \}. 
\end{align*}
We denote by $Y$ the set of vertices that do not belong to these subsets, i.e., $Y=V^*\setminus (R\cup Z)$. 

For each vertex $v\in R$, we update $p(v)$ as 
$$p(v):= \begin{cases}
 p(v)+\epsilon & (v\in R^+\cap B^*) \\ 
 p(v)-\epsilon & (v\in R^+\setminus B^*) \\
 p(v)-\epsilon & (v\in R^-\cap B^*) \\
 p(v)+\epsilon & (v\in R^-\setminus B^*). 
\end{cases}$$
We also modify $q(H)$ for each maximal blossom $H$ by  
$$q(H):=
\begin{cases}
 q(H)+\epsilon & (H: \mbox{labeled with $\oplus$}) \\
 q(H)-\epsilon & (H: \mbox{labeled with $\ominus$}) \\
 q(H)          & (\mbox{otherwise}).  
\end{cases}$$
To keep the feasibility of the dual variables, $\epsilon$ is determined by 
$\epsilon=\min\{\epsilon_1,\epsilon_2,\epsilon_3,\epsilon_4\}$, where 
\begin{align*}
\epsilon_1 & =  \frac{1}{2}\min\{p(v)-p(u)-Q_{uv}\mid (u, v)\in F^*,\, u, v\in R^+\cup Z^+,\, K(u)\not=K(v) \}, \\
\epsilon_2 & =  \min\{p(v)-p(u)-Q_{uv}\mid (u, v)\in F^*,\, u\in R^+\cup Z^+,\,v\in Y \}, \\ 
\epsilon_3 & =  \min\{p(v)-p(u)-Q_{uv}\mid (u, v)\in F^*,\, u\in Y,\,v\in R^+\cup Z^+ \}, \\
\epsilon_4 & =  \min\{q(H)\mid \mbox{$H$: a maximal blossom labeled with $\ominus$}\}.
\end{align*}
If $\epsilon = + \infty$, then we terminate $\Search$ and conclude that 
there exists no parity base. 
Otherwise, 
while there exists a maximal blossom whose value of $q$ is zero after the dual update, 
delete such a blossom from $\Lambda$
by $\Expand$. Then, apply the procedure $\Search$ again.

\subsection{Detecting Infeasibility}
\label{sec:infeasible}

By the definition of $\epsilon$, 
we can easily see that the updated dual variables are feasible 
if $\epsilon$ is a finite value. 
We now show that we can conclude that the instance has no parity base if $\epsilon = + \infty$. 

A skew-symmetric matrix is called an {\em alternating matrix} 
if all the diagonal entries are zero. 
Note that any skew-symmetric matrix is alternating 
unless the underlying field is of characteristic two.
By a congruence transformation, 
an alternating matrix can be brought into a block-diagonal form 
in which each nonzero block is a $2 \times 2$ alternating matrix. 
This shows that the rank of an alternating matrix is even, 
which plays an important role in the proof of the following lemma. 

\begin{lemma}
\label{lem:noparitybase}
Suppose that there is a source line, and  
suppose also that $\epsilon =  + \infty$
when we update the dual variables.
Then, the instance has no parity base. 
\end{lemma}

\begin{proof}
In order to show that there is no parity base, by Lemma~\ref{lem:Pf}, it suffices to show that $\Pf\Phi_A(\theta)=0$. 
We construct the matrix 
$$\Phi_A^*(\theta)=
\left( 
\begin{array}{c|c|c|c|c}
\multicolumn{2}{c|}{}  &  O & \multicolumn{2}{c}{}   \\ \cline{3-3}
\multicolumn{2}{c|}{\raisebox{7pt}[0pt][0pt]{$O$}} & I  & \multicolumn{2}{c}{\raisebox{7pt}[0pt][0pt]{\quad $C^*$}\quad}   \\ \hline
 O & -I  & \multicolumn{2}{c|}{} & \\ \cline{1-2}
\multicolumn{2}{c|}{} & \multicolumn{2}{c|}{\raisebox{7pt}[0pt][0pt]{$D'(\theta)$}} & \raisebox{7pt}[0pt][0pt]{$O$} \\ \cline{3-5}
\multicolumn{2}{c|}{\raisebox{7pt}[0pt][0pt]{$-{C^*}^\top$}} & \multicolumn{2}{c|}{O} & O  
\end{array}
\right)
\begin{array}{l}
\leftarrow T \cap B^* \\ 
\leftarrow U \mbox{ (identified with $B$)} \\
\leftarrow B \\ 
\leftarrow V\setminus B \\
\leftarrow T \setminus B^* 
\end{array}
$$
in the same way as Section~\ref{sec:optimality}, where $T := \{b_i, t_i \mid H_i \in \Lambda_{\rm n} \}$.  
Note that we regard the row set of $C^*$ as $(T \cap B^*) \cup U$ instead of $U^*$, 
and hence the row/column set of $\Phi^*_A(\theta)$ is $W^* := V^* \cup U$. 
Then $\Pf\Phi_A(\theta)=0$ is equivalent to $\Pf\Phi_A^*(\theta) =0$. 

Construct a graph $\Gamma^*=(W^*,E^*)$ with edge set $E^*:=\{(u,v) \mid (\Phi^*_A(\theta))_{u v} \neq 0 \}$.
In order to show that $\Pf\Phi_A^*(\theta)=0$, it suffices to prove that $\Gamma^*$ does not have a perfect matching. 
Since $\Phi^*_A(\theta) [U, B]$ is the identity matrix, we have a natural bijection 
$\eta: B \to U$ between $B$ and $U$. We then define $X\subseteq W^*$ by 
$X:=(R^-\setminus B)\cup\eta(R^-\cap B)$. 

Since $\epsilon_4 = + \infty$, no maximal blossom $H_i$ is labeled with $\ominus$. 
For each maximal blossom $H_i$ labeled with $\oplus$, we introduce $Z_i := H_i \cup \eta(H_i \cap B)$.
If $H_i$ is a normal blossom, then $H_i$ is of odd cardinality and $H_i$ does not contain any source line, 
which imply that $|Z_i|$ is odd.   
If $H_i$ is a source blossom, then $H_i$ is of even cardinality and $H_i$ contains exactly one source line, 
which again imply that $Z_i$ is of odd cardinality. Note that there exist no edges of $E^*$ between $Z_i$ 
and $W^*\setminus (X\cup Z_i)$. 
 
All the source lines that are not included in any blossoms are contained in $R^+$.
For each normal line $\ell \subseteq R$, exactly one vertex $u_\ell$ in $\ell$ is 
unlabeled and the other vertex $\bar u_\ell$ is labeled. For each line $\ell\subseteq R$, 
we now introduce $R_\ell$ by 
$$R_\ell:=
\begin{cases}
 \{u_\ell, \bar u_\ell, \eta(\bar u_\ell)\} & (\ell\subseteq B), \\
 \{v_\ell, \bar v_\ell, \eta(\bar v_\ell)\} & (\ell=\{v_\ell,\bar v_\ell\}, \bar v_\ell\in B, v_\ell\in V\setminus B), \\
 \{\bar u_\ell\}  & (\ell\subseteq V\setminus B).  
\end{cases}$$
Note that $R_\ell$ is of odd cardinality and that there exist no edges of $E^*$ between $R_\ell$ and $W^*\setminus (X\cup Z_i)$. 

Let $\odd(\Gamma^*\setminus X)$ denote the number of odd components after deleting $X$ from $\Gamma^*$. 
For each $b_i\in R^-$, we have a corresponding odd component $Z_i$. For each $u_\ell\in R^-$, we have an odd
component $R_\ell$. In addition, there are some other odd components coming from source blossoms or source lines. 
Thus we have $\odd(\Gamma^*\setminus X)>|X|$, which implies by the theorem of Tutte~\cite{Tut47} that 
$\Gamma^*$ does not admit a perfect matching. 
\end{proof}

\subsection{Bounding Iterations}
\label{sec:iterations}

We next show that the dual variables are updated $O(n)$ times per augmentation. 
To see this, roughly, 
we show that this operation increases the number of labeled vertices. 
Although $\Search$ contains flexibility on the ordering of vertices, 
it does not affect the set of the labeled vertices when $\Search$ returns $\emptyset$.  
This is guaranteed by the following lemma.

\begin{lemma}
\label{lem:labeliff}
Suppose that a vertex $v \in V \cup \{b_i \mid \mbox{$H_i \in \Lambda_{\rm n}$ is a maximal blossom} \}$ is not removed in $\Search$
that returns $\emptyset$. 
Then, $v$ is labeled in $\Search$
if and only if 
there exists a vertex set $X \subseteq V^*$ such that 
\begin{itemize}
\item
$X \cup \{v\}$ consists of normal lines, dummy lines, and a source vertex $s$, 
\item
$T \subseteq X \cup \{v\}$,   
\item
$C^*[X]$ is nonsingular, and
\item
the equality  
\begin{equation}
\label{eq:8411}
p(X \setminus B^*) - p (X \cap B^*)  
= \sum \{ q( H_i) \mid H_i \in \Lambda,\ \mbox{$|X\cap H_i|$ is odd}\} 
\end{equation}
holds.
\end{itemize}
\end{lemma}

\begin{proof}
We first observe that 
creating or deleting a blossom does not affect the conditions in Lemma~\ref{lem:labeliff} unless $v$ is removed. 
Indeed, when $T$ is updated as $T' := T \cup \{b_i, t_i\}$ or $T' := T \setminus \{b_i, t_i\}$ 
by creating/deleting a blossom, 
$X' := ((X \setminus T) \cup T') \setminus \{v\}$
satisfies the conditions by Lemma~\ref{lem:pivotsing}. 
Thus, it suffices to show that 
$v$ is labeled in $\Search$
if and only if 
there exists a vertex set $X$ satisfying the conditions when $\Search$ returns $\emptyset$.
In what follows in the proof, all notations ($V^*, C^*, T, \Lambda$, etc.) represent the objects when $\Search$ returns $\emptyset$.

If $v$ is labeled 
in $\Search$, 
then we obtain $P(v)$ such that 
$G^\circ[P(v) \setminus \{v\}]$ has a unique tight perfect matching by Corollary~\ref{cor:reachable}. 
Define $X := (P(v) \cup T) \setminus \{v\}$. 
For any minimal $H_i \in \Lambda_{\rm n}$ with $P(v) \cap H_i = \emptyset$, 
it follows from Lemma~\ref{clm:71} that $(b_i, t_i)$ is a unique edge in $G^\circ$ between $t_i$ and $X \setminus \{t_i\}$. 
Thus, if $G^\circ[X]$ has a perfect matching, then it must contain $(b_i, t_i)$. 
By applying this argument repeatedly for each $H_i \in \Lambda_{\rm n}$ with $P(v) \cap H_i = \emptyset$
in the order of indices 
(i.e., in the order from smaller blossoms to larger ones), 
$G^\circ[X]$ has a unique tight perfect matching, 
because $b_i, t_i \in P(v)$ for any $H_i \in \Lambda_{\rm n}$ with $P(v) \cap H_i \not= \emptyset$ by Observation~\ref{obs:budtipinP}. 
Thus, $C^*[X]$ is nonsingular by Lemma~\ref{lem:tightmatching}, and the equality (\ref{eq:8411}) holds. 

We now intend to prove the converse. 
Suppose that $X$ satisfies the above conditions, 
and assume to the contrary that $v$ is not labeled 
when $\Search$ returns $\emptyset$. Then, we can update the dual variables keeping the dual feasibility as described at the beginning of this section. 
We now see how the dual update affects (\ref{eq:8411}). 
\begin{itemize}
\item
Consider the dual variables corresponding to $K(s)$. 
If $s$ is single, then
the left hand side of (\ref{eq:8411}) decreases by $\epsilon$ by updating $p(s)$. 
Otherwise, 
$K(s) = H_i$ for some source blossom $H_i \in \Lambda_{\rm s}$, since $s$ is a source vertex. 
Then, $|X \cap H_i|$ is odd as $v \not\in H_i$, and hence
the right hand side of (\ref{eq:8411}) increases by $\epsilon$ by updating $q(H_i)$.
\item
Consider the dual variables corresponding to $K(v)$. 
\begin{itemize}
\item
If $v$ is single, then
the left hand side of (\ref{eq:8411}) decreases by $\epsilon$ or does not change by updating $p(\bar v)$, 
because $\bar v \in R^+ \cup Y$. 
\item
If $v \in H_i$ for some maximal blossom $H_i \in \Lambda_{\rm n}$, then $|X \cap H_i|$ is even. 
Thus, the right hand side of (\ref{eq:8411}) does not change by updating $q(H_i)$. 
Furthermore, since $H_i$ is not labeled with $\oplus$, 
we have $b_i \in R^+ \cup Y$, which shows that 
the left hand side of (\ref{eq:8411}) decreases by $\epsilon$ or does not change by updating $p(b_i)$. 
\item
If $v = b_i$ for some maximal blossom $H_i \in \Lambda_{\rm n}$, then $|X \cap H_i|$ is odd. 
Since $H_i$ is not labeled with $\ominus$, 
the right hand side of (\ref{eq:8411}) increases by $\epsilon$ or does not change by updating $q(H_i)$.
\item
If $v\in H_i$ for some maximal blossom $H_i\in\Lambda_{\rm s}$, then 
$v$ is labeled, which contradicts the assumption. 
\end{itemize}
\item 
For any $u \in X$ with $s, v \not\in K(u)$, 
updating the dual variables corresponding to $K(u)$ does not 
 affect the equality (\ref{eq:8411}), 
since $|X\cap H_i|$ is even for any $H_i \in \Lambda_{\rm s}$ with $s \not\in H_i$ 
and $|X \cap H_i|$ is odd for any $H_i \in \Lambda_{\rm n}$ with $v \not\in H_i \cup \{b_i\}$. 
\end{itemize}
By combining these facts, after updating the dual variables, we have that 
the left hand side of (\ref{eq:8411}) is strictly less than its right hand side, 
which contradicts Lemma~\ref{lem:keyodd}. 
\end{proof}

By using this lemma, we bound the number of dual updates as follows.

\begin{lemma}
\label{lem:dualupdatebound}
The dual variables are updated at most $O(n)$ times
before $\Search$ finds an augmenting path or 
we conclude that the instance has no parity base by Lemma~\ref{lem:noparitybase}. 
\end{lemma}

\begin{proof}
Suppose that we update the dual variables more than once, and 
we consider how the value of 
$$\kappa(V^*,\Lambda):=|\{ w \in V \mid \mbox{$w$ is labeled} \}| + |\Lambda_1| -  |\Lambda_2| -  2 |\Lambda_3|$$
will change between two consecutive dual updates, where
\begin{align*}
\Lambda_1&:= 
|\{H_i \in \Lambda \mid \mbox{$H_i$ contains a labeled vertex} \}|,  \\
\Lambda_2&:= 
|\{H_i \in \Lambda_{\rm n} \mid \mbox{$H_i$ is a maximal blossom labeled with $\ominus$} \}|, \\
\Lambda_3&:= \Lambda \setminus (\Lambda_1 \cup \Lambda_2). 
\end{align*}
Note that every maximal blossom labeled with $\ominus$ contains no labeled vertex, and hence $\Lambda_1 \cap \Lambda_2 = \emptyset$.
We first show that $\kappa(V^*,\Lambda)$ does not decrease.
 
By Lemma~\ref{lem:labeliff}, if $w \in V$ is labeled at the time of the first dual update, 
then it is labeled again at the time of the second dual update. 
This shows that $|\{ w \in V \mid \mbox{$w$ is labeled} \}|$ does not decrease.
By Lemma~\ref{lem:labeliff} again, blossoms satisfy the following.
\begin{itemize}
\item
If a blossom is in $\Lambda_1$ at the time of the first dual update, then 
it is still in $\Lambda_1$ at the time of the second dual update unless it is deleted.
Note that such a blossom is deleted only when it is replaced with a new blossom in $\DBlossom$.
\item
If a blossom is in $\Lambda_2$ at the time of the first dual update, then 
it is in $\Lambda_1 \cup \Lambda_2$ at the time of the second dual update unless it is deleted.
\item
If a blossom is in $\Lambda_3$ at the time of the first dual update, then 
it is in $\Lambda = \Lambda_1 \cup \Lambda_2 \cup \Lambda_3$ at the time of the second dual update unless it is deleted.
\item
If a new blossom is created in $\Blossom$ after the first dual update, then it is in $\Lambda_1$ at the time of the second dual update.
\item
If $\DBlossom$ is applied after the first dual update, then it replaces a blossom in $\Lambda$ with 
a new blossom containing a labeled vertex, i.e., the new blossom is in $\Lambda_1$ at the time of the second dual update. 
\end{itemize}
By the above observations, $\kappa(V^*,\Lambda)$ does not decrease. 
In what follows, we show that $\kappa(V^*,\Lambda)$ increases strictly.

If we update the dual variables with $\epsilon = \epsilon_4$, then 
there exists a maximal blossom 
$H_i \in \Lambda_{\rm n}$ labeled with $\ominus$ such that $q(H_i) = \epsilon$, 
which shows that 
$H_i \in \Lambda_2$ is deleted before the time of the second dual update.
This shows that $\kappa(V^*,\Lambda)$ increases. 

If $\epsilon < \epsilon_4$, then
there is a new tight edge between $R^+\cup Z^+$ and $Y$, or between two vertices in $R^+\cup Z^+$. 
We note that some blossoms may be created or deleted in $\DBlossom$ after the first dual update is executed. 
However, such a new tight edge remains to exist by Lemmas~\ref{clm:71} and~\ref{lem:dblossomexpand}. 

Suppose that $\epsilon = \epsilon_2$. In this case, 
we create a new tight edge $(u, v)$ with $u\in R^+\cup Z^+$ and $v\in Y$. 
Since $u$ is labeled again at the time of the second dual update, 
some vertex in $K(v)$ is newly labeled. 
Thus, 
$|\{ w \in V \mid \mbox{$w$ is labeled} \}|$ increases or a blossom in $\Lambda_3$ becomes a member of $\Lambda_2$,  
and hence
the value of $\kappa(V^*,\Lambda)$ will increase. 
The same argument can be applied to the case of $\epsilon = \epsilon_3$. 

Suppose that $\epsilon = \epsilon_1$. In this case, 
we create a new tight edge $(u, v)$ with $u, v\in R^+\cup Z^+$ and $K(u) \not= K(v)$. 
By changing the roles of $u$ and $v$ if necessary, we may assume that $u \prec v$. 
Then, we consider each of the following cases. 
\begin{itemize}
\item
If the first elements in $P(v)$ and $P(u)$ belong to different source lines, then
we obtain an augmenting path, which contradicts that we apply the second dual update. 
\item
If $v \in H_i$ for some maximal normal blossom $H_i \in \Lambda_{\rm n}$ and $u = \rho (b_i)$, then 
there exists an edge in $F^*$ between $u=\rho(b_i)$ and $v \in H_i$, 
which contradicts Lemma~\ref{lem:72}. 
\item
If neither of the above cases apply, then a new blossom $H$ is created in $\Blossom(v, u)$, 
and hence $|\Lambda_1|$ increases.
This shows that the value of $\kappa(V^*,\Lambda)$ increases. 
\end{itemize}

Thus, the value of $\kappa(V^*,\Lambda)$ increases by at least one between two consecutive dual updates. 
Since the range of $\kappa(V^*,\Lambda)$ is at most $O(n)$, 
the dual variables are updated at most $O(n)$ times.
\end{proof}

\section{Augmentation}
\label{sec:augmentation}

The objective of this section is to describe how to update the primal solution using an augmenting path $P$. 
The augmentation procedure that primarily replaces $B^*$ with $B^* \triangle P$, where $\triangle$ denotes 
the symmetric difference. In addition, it updates the bud and the tip of each normal blossom.

Suppose we are given $V^*$, $B^*$, $C^*$, $\Lambda$, and feasible dual variables 
$p$ and $q$. 
Let $P$ be an augmenting path, and 
$\Lambda_P$ denote the set of blossoms 
that intersect with $P$, i.e., $\Lambda_P=\{H_i\in \Lambda\mid H_i\cap P\neq \emptyset\}$.
Let $\Lambda^+_P$ denote the set of positive blossoms in $\Lambda_P$. 
In the augmentation along $P$, we update 
$V^*$, $B^*$, $C^*$, 
$\Lambda$, $b_i$, $t_i$, 
$p$, and $q$.
The procedure for augmentation is described as follows.

\begin{description}

\item[Procedure $\Augment(P)$]

\item[Step 0:]
While there exists a maximal blossom $H_i \in \Lambda \setminus \Lambda_P$ with $q(H_i) = 0$, apply $\Expand(H_i)$.

\item[Step 1:]
Let $M$ be the unique tight perfect matching in $G^\circ[P]$. 
For each $H_i\in\Lambda^+_P$, 
let $(x_i,y_i)$ be the unique edge in $M$ with $x_i\in H_i$ and $y_i\in V^*\setminus H_i$ (see Corollary~\ref{cor:uniqueedgeaug}),
add new vertices $\hat b_i$ and $\hat t_i$ to $V^*$, and update $B^*$, $C^*$, and $p$ as follows (see Fig.~\ref{fig:revisionfig32}). 
\begin{itemize}
\item Add $\hat t_i$ to $H_i$. For each blossom $H_j$ with $H_i\subsetneq H_j$, add $\hat b_i$ and $\hat t_i$ to $H_j$.  
\item If $x_i\in B^*$ and $y_i\in V^*\setminus B^*$, then $B^*:=B^*\cup\{\hat b_i\}$,
$$C^*_{\hat b_i v}:=
\begin{cases} 
C^*_{x_i v}  & (v\in (V^*\setminus B^*)\setminus H_i), \\
0 & (v\in H_i\setminus B^*), 
\end{cases}
\quad\quad
C^*_{u \hat t_i}:=
\begin{cases}
C^*_{u y_i} & (u\in B^*\cap H_i), \\
0 & (u\in B^*\setminus H_i), 
\end{cases}
$$
$p(\hat b_i):=p(y_i) - Q_{\hat b_i y_i}$, and $p(\hat t_i):=p(x_i)+Q_{x_i\hat t_i}$.
\item If $x_i\in V^*\setminus B^*$ and $y_i\in B^*$, then $B^*:=B^*\cup\{\hat t_i\}$,
$$C^*_{u \hat b_i}:=
\begin{cases} 
C^*_{u x_i}  & (u\in B^*\setminus H_i), \\
0 & (u \in B^*\cap H_i), 
\end{cases}
\quad\quad
C^*_{\hat t_i v}:=
\begin{cases}
C^*_{y_i v} & (v\in H_i\setminus B^*), \\
0 & (v \in (V^*\setminus B^*)\setminus H_i),  
\end{cases}
$$
$p(\hat b_i):=p(y_i) + Q_{\hat b_i y_i}$, and $p(\hat t_i):=p(x_i) - Q_{x_i\hat t_i}$.
\end{itemize}

\begin{figure}[htbp]
  \centering
    \includegraphics[width=8.5cm]{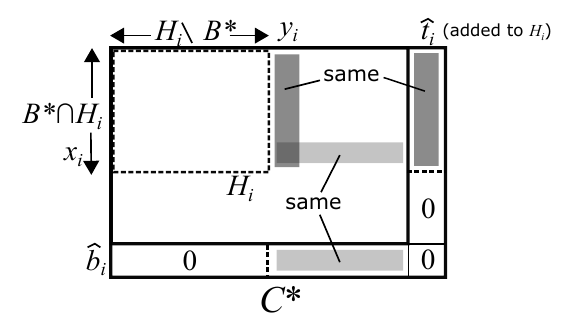}
      \caption{Definition of $C^*$ in $\Augment(P)$.} 
    \label{fig:revisionfig32} 
\end{figure}

\item[Step 2:] Apply the pivoting operation around $P^*:=P\cup \{\hat b_i, \hat t_i\mid H_i\in \Lambda^+_P\}$ to $C^*$, 
namely $B^*:=B^*\triangle P^*$.

\item[Step 3:] 
For each (not necessarily maximal) blossom $H_i\in\Lambda_P \setminus \Lambda^+_P$, remove $H_i$ from $\Lambda$, 
and if $H_i$ is a normal blossom, then remove also $b_i$ and $t_i$ from $V^*$. 
For each $H_i\in\Lambda^+_P$, remove $b_i$ and $t_i$ from $V^*$ if $H_i$ is a normal blossom, and rename $\hat b_i$ and $\hat t_i$ 
as the bud $b_i$ and the tip $t_i$ of $H_i$, respectively.

\item[Step 4:]
For each $H_i\in\Lambda^+_P$ in the order of indices 
(i.e., in the order from smaller blossoms to larger ones),
apply the following. 
\begin{enumerate}
\item[(i)]
Introduce new vertices $b'_i$ and $t'_i$ and
add $t'_i$ to $H_i$. 
For each blossom $H_j$ with $H_i\subsetneq H_j$, add $b'_i$ and $t'_i$ to $H_j$.  
\item[(ii)] 
If $b_i\in B^*$ and $t_i\in V^*\setminus B^*$, then  
$B^*:=B^*\cup\{t'_i\}$, 
$$C^*_{u b'_i}:=
\begin{cases} 
C^*_{u t_i}  & (u \in B^* \setminus H_i), \\
0 & (u \in H_i \cap B^*), 
\end{cases}
\quad\quad
C^*_{t'_i v}:=
\begin{cases}
C^*_{b_i v} & (v\in H_i \setminus B^*), \\
0 & (v\in (V^*\setminus B^*)\setminus H_i), 
\end{cases}
$$
$p(b'_i):=p(t_i) - Q_{b'_i t_i}$, and $p(t'_i):=p(b_i) + Q_{b_i t'_i}$.

\item[(iii)]
If $b_i\in V^* \setminus B^*$ and $t_i\in B^*$, then  
$B^*:=B^*\cup\{b'_i\}$, $$C^*_{b'_i v}:=
\begin{cases} 
C^*_{t_i v}  & (v \in (V^* \setminus B^*) \setminus H_i), \\
0 & (v \in H_i \setminus B^*), 
\end{cases}
\quad\quad
C^*_{u t'_i}:=
\begin{cases}
C^*_{u b_i} & (u\in H_i \cap B^*), \\
0 & (u\in B^* \setminus H_i), 
\end{cases}
$$
$p(b'_i):=p(t_i) + Q_{b'_i t_i}$, and $p(t'_i):=p(b_i) - Q_{b_i t'_i}$.

\item[(iv)] Apply the pivoting operation around $\{b_i, t_i, b'_i, t'_i\}$ to $C^*$, namely $B^*:=B^*\triangle\{b_i, t_i, b'_i, t'_i\}$. 

\end{enumerate}

Then, for each $H_i\in\Lambda^+_P$, remove $b_i$ and $t_i$ from $V^*$, and rename $b'_i$ and $t'_i$ 
as the bud $b_i$ and the tip $t_i$ of $H_i$, respectively.

\item[Step 5:]
For each $H_i\in\Lambda^+_P$ in the reverse order of indices
(i.e., in the order from larger blossoms to smaller ones),
apply the procedures (i)--(iv) in Step 4. 
Then, for each $H_i\in\Lambda^+_P$, remove $b_i$ and $t_i$ from $V^*$, and rename $b'_i$ and $t'_i$ 
as the bud $b_i$ and the tip $t_i$ of $H_i$, respectively. 
\end{description}

Note that Steps 4 and 5 are executed to keep (BT2). 
After Step 3, (BT2) does not necessarily hold, whereas the dual variables are feasible and (BT1) holds. 
Step 4 is applied to delete all the edges in 
$F^*$ between $t_i$ and $(V^* \setminus H_i) \setminus \{b_i\}$ for each $H_i \in \Lambda^+_{P}$, and 
Step 5 is applied to delete all the edges in 
$F^*$ between $b_i$ and $H_i \setminus \{t_i\}$ for each $H_i \in \Lambda^+_{P}$.
See Lemma~\ref{lem:augBT} for details.

In Section~\ref{sec:augmentfeas}, 
we show the validity of the augmentation procedure. 
After the augmentation, the algorithm applies $\Search$ in each blossom $H_i$
to obtain a new routing and ordering in $H_i$, 
which will be described in Section~\ref{sec:reroute}.

\subsection{Validity}
\label{sec:augmentfeas}

In this subsection, we show the validity of $\Augment(P)$.
We first show that 
the dual feasibility holds after the augmentation.

\begin{lemma}
Suppose that the dual variables $(p,q)$ are feasible at the beginning of $\Augment(P)$. 
Then the procedure keeps the dual feasibility. 
\end{lemma}
\begin{proof}
By Lemma~\ref{lem:expand}, the dual variables $(p,q)$ are feasible after Step 0. 

We intend to show that $(p,q)$ are feasible after Step 1. 
New edges that appear in $F^*$ are incident to $\hat b_i$ or $\hat t_i$
for some $H_i\in\Lambda_P$. For a new edge $(u,\hat{t_i})\in F^*$, we have $u\in H_i$, $(u,y_i)\in F^*$,  
and $Q_{u y_i}-Q_{u \hat t_i} = Q_{x_i y_i} - Q_{x_i \hat t_i}$. If $x_i\in B^*$, we have 
\begin{eqnarray*}
p(\hat t_i)-p(u) & = & p(x_i)+Q_{x_i\hat t_i}-p(u) \\
                 & = & p(y_i)-Q_{x_iy_i}+Q_{x_i\hat t_i} - p(u) \\ 
                 & = & p(y_i)-Q_{uy_i}+Q_{u\hat t_i} -p(u) \geq Q_{u\hat t_i}. 
\end{eqnarray*}
If $x_i\in V^*\setminus B^*$, we can similarly derive $p(u) - p(\hat t_i) \geq Q_{u\hat t_i}$. 
For a new edge $(\hat{b_i},v)\in F^*$, we have $v\in V^*\setminus H_i$, $(x_i,v)\in F^*$,  
and $Q_{x_i v}-Q_{\hat b_i v} = Q_{x_i y_i} - Q_{\hat b_i y_i}$. If $x_i\in B^*$, 
we have 
\begin{eqnarray*}
p(v) - p(\hat b_i) & = & p(v) - p(y_i) + Q_{\hat b_i y_i}  \\
                   & = & p(v) - p(x_i) - Q_{x_iy_i} + Q_{\hat b_i y_i} \\ 
                   & = & p(v) - p(x_i) - Q_{x_i v} + Q_{\hat b_i v}  \geq Q_{x_i v}. 
\end{eqnarray*}
If $x_i\in V^*\setminus B^*$, we can similarly derive $p(\hat b_i) - p(v) \geq Q_{\hat b_i v}$. 
Thus the dual variables $(p,q)$ remain feasible at the end of Step 1. 

We next intend to show that Step 2 also keeps the dual feasibility. Suppose that $(u,v)\in F^*$ with 
$u\in B^*$ and $v\in V^*\setminus B^*$ after Step 2. Then $C^*[P^*\triangle\{u,v\}]$ must be nonsingular 
before the pivoting operation by Lemma~\ref{lem:pivotsing}. 
Since $|P^*\cap H_i|$ is even for each $H_i\in\Lambda$ with $q(H_i) > 0$, it follows from 
Lemma~\ref{lem:keyodd} that  
$$p((P^*\triangle\{u,v\})\setminus B^*)-p((P^*\triangle\{u,v\})\cap B^*) \geq  Q_{uv}$$
before Step 2. 
On the other hand, since $G^\circ[P^*]$ contains a tight perfect matching, we have 
$$p(P\setminus B^*)-p(P^*\cap B^*) =0$$
before Step 2. 
Combining these two inequalities with $u\in P^*\triangle B^*$ and $v\in V^*\setminus (P^*\triangle B^*)$, 
we obtain $p(v)-p(u)\geq Q_{uv}$, which shows that $(p,q)$ remain feasible after Step 2. 

Removing some vertices in Step 3 does not affect the dual feasibility. 

Finally, we consider each step of Steps 4 and 5. 
We can see that adding $b'_i$ and $t'_i$ does not violate the dual feasibility by the same argument as Step 1. 
If $(u, v) \in F^*$ after the pivoting operation in Step 4 or 5, then 
$C^*[X_i \triangle \{u, v\}]$ is nonsingular where $X_i := \{b_i, t_i, b'_i, t'_i\}$ 
before the pivoting operation by Lemma~\ref{lem:pivotsing}. 
Since $G^\circ[X_i]$ contains a tight perfect matching before the pivoting operation, 
we can apply the same argument as Step 2 to show that 
$(p,q)$ remain feasible after Steps 4 and 5. 

Thus $(p,q)$ is feasible throughout the procedure.  
\end{proof}

We next show the nonsingularity of $C^*[P^*]$ in Step 2, 
which guarantees that we can apply the  pivoting operation in Step 2 of $\Augment(P)$. 

\begin{lemma}
When we apply the pivoting operation in Step 2 of $\Augment(P)$, 
$C^*[P^*]$ is nonsingular.  
\end{lemma}
\begin{proof}
We first note that  
$\Expand(H_i)$ in Step 0 does not affect the edges in $G^\circ[P]$. 

We show that  
$G^\circ[P']$ has a unique tight perfect matching for $P' := P \cup \{\hat b_i, \hat t_i\}$
with $H_i \in \Lambda^+_P$. 
Since $G^\circ[P]$ has a unique tight perfect matching $M$, which contains $(x_i, y_i)$, 
both $G^\circ[(P \cap H_i) \cup \{y_i\}]$
and $G^\circ[(P \setminus H_i) \cup \{x_i\}]$	
have a unique tight perfect matching. 
By the definition of $\hat b_i$ and $\hat t_i$, this shows that
both $G^\circ[P' \cap H_i]$
and $G^\circ[P' \setminus H_i]$
have a unique tight perfect matching. 
Thus, we obtain a tight perfect matching in $G^\circ[P']$. 
Furthermore, since $|H_i \cap P'|$ is even and $H_i$ is positive, 
any tight perfect matching in $G^\circ[P']$ 
consists of a tight perfect matching in $G^\circ[P' \cap H_i]$ and one in $G^\circ[P' \setminus H_i]$. 
Therefore, $G^\circ[P']$ has a unique tight perfect matching. 

By applying the same argument to each $H_i \in \Lambda^+_P$, repeatedly, 
we see that $G^\circ[P^*]$ has a unique tight perfect matching. 
By Lemma~\ref{lem:tightmatching}, 
$G^*[P^*]$ has a unique perfect matching, 
which shows that $C^*[P^*]$ is nonsingular.
\end{proof}

Finally in this subsection, we show that 
(BT1) and (BT2) hold after $\Augment(P)$.  

\begin{lemma}\label{lem:augBT}
The procedure $\Augment(P)$ keeps {\em (BT1)} and {\em (BT2)}. 
\end{lemma}

\begin{proof}
It is obvious from the definition that (BT1) holds. 

We first show by induction on $i$ that, for any $j \le i$ with $H_j\in\Lambda^+_P$,  
$(b'_j, t'_j) \in F^*$ and 
there is no edge in $F^*$ between $t'_j$ and $(V^* \setminus H_j) \setminus \{b'_j\}$ 
after the pivoting operation around $X_i := \{b_i, t_i, b'_i, t'_i\}$ in Step 4. 
We only consider the case when 
$b'_i \in B^*$ and $t'_i \in V^* \setminus B^*$ after the pivoting operation
as the other case can be dealt with in a similar way. 
Since 
$$
C^*\left[ P^*  \setminus \{\hat b_j, \hat t_j \mid H_j \in \Lambda^+_P, j\le i\}\right]
$$
is nonsingular before the pivoting operation around $P^*$ in Step 2, we have 
$C^*[\{b_j, t_j \mid H_j \in \Lambda^+_{P}, j \le i\}]$ is nonsingular after the pivoting operation around $P^*$ in Step 2 by Lemma~\ref{lem:pivotsing}. 
By Lemma~\ref{lem:pivotsing} again, 
this shows that 
$C^*[\{b'_j, t'_j \mid H_j \in \Lambda^+_{P}, j \le i\}]$ is nonsingular 
after the pivoting operation around $X_i$ in Step 4. 
Since there is no edge between $t'_j$ and $(V^* \setminus H_j) \setminus \{b'_j\}$ for $j < i$ by induction hypothesis, 
the nonsingularity of $C^*[\{b'_j, t'_j \mid H_j \in \Lambda^+_{P}, j \le i\}]$ shows that $C^*_{b'_i t'_i} \not= 0$.   
Before the pivoting operation around $X_i$, for $u \in B^*\setminus H_i$ with $u\neq b_i$, 
$\det C^*[X_i \triangle \{t'_i, u\}] = \det C^*[\{b_i, t_i, b'_i, u\}]$ 
is zero, since two columns in $C^*[\{b_i, t_i, b'_i, v\}]$ are the same by the definition of $b'_i$. 
Thus, $C^*_{u t'_i} = 0$ for $u \in B^*\setminus H_i$ with $u\neq b'_i$
after the pivoting operation around $X_i$. 
Furthermore, for any $j < i$ with $H_j \in \Lambda^+_P$, the pivoting operation around $X_i$ does not create 
a new edge in $F^*$ between $t'_j$ and $v \in (V^* \setminus H_j) \setminus \{b'_j\}$, 
because a row/column of $C^*[X_i \triangle \{t'_j, v\}]$ corresponding to $v$ is zero before the pivoting operation around $X_i$.
We can also see that 
the pivoting operation around $X_i$ does not remove $(b'_j, t'_j)$ from $F^*$ for any $j < i$ with $H_j \in \Lambda^+_P$.  
Hence, for each $H_i \in \Lambda^+_P$, $(b'_i, t'_i) \in F^*$ and 
there is no edge in $F^*$  between $t'_i$ and $(V^* \setminus H_i) \setminus \{b'_i\}$ 
after applying (i)--(iv) for each normal blossom in Step 4. 

We next show by induction on $i$ (in the reverse order) that, for any $j \ge i$ with $H_j \in \Lambda^+_P$,  
$H_j$ satisfies the condition in (BT2) after the pivoting operation around $X_i := \{b_i, t_i, b'_i, t'_i\}$ in Step 5. 
Note that the pivoting operation around $X_i$ creates/deletes neither  
an edge in $F^*$ between $t'_j$ and $(V^* \setminus H_j) \setminus \{b'_j\}$ for $j \not= i$, 
nor 
an edge in $F^*$ between $b'_j$ and $H_j \setminus \{t'_j\}$ for $j > i$. 
Thus, it suffices to show that there is no edge in $F^*$ between $b'_i$ and $H_i \setminus \{t'_i\}$ 
after the pivoting operation around $X_i$ in Step 5.   
We only consider the case when 
$b'_i \in B^*$ and $t'_i \in V^* \setminus B^*$ after the pivoting operation
as the other case can be dealt with in a similar way. 
Before the pivoting operation around $X_i$, for $v \in H_i\cap B^*$ with $v\neq t_i$, 
$\det C^*[X_i \triangle \{b'_i, v\}] = \det C^*[\{b_i, t_i, t'_i, v\}]$ 
is zero, since two rows in $C^*[\{b_i, t_i, t'_i, v\}]$ are the same by the definition of $t'_i$. 
Thus, $C^*_{b'_i v} = 0$ for $v \in H_i\cap B^*$ with $v\neq t'_i$
after the pivoting operation around $X_i$. 
Hence, by applying (i)--(iv) for each normal blossom in Step 5, 
there is no edge in $F^*$ between $b'_i$ and $H_i \setminus \{t'_i\}$ for each $H_i \in \Lambda^+_P$. 

Since the pivoting operations do not 
create/delete an edge in $F^*$ between $t'_i$ and $(V^* \setminus H_i) \setminus \{b'_i\}$ for each $H_i \in \Lambda_{\rm n} \setminus \Lambda^+_P$, 
(BT2) holds after $\Augment(P)$.   
\end{proof}

\subsection{Search in Each Blossom}
\label{sec:reroute}
In this subsection, we describe how to update the routing $R_{H_i}(x)$ for each $x\in H^\bullet_i$ and the ordering 
$<_{H_i}$ in $H^\bullet_i$ after the augmentation. 
If $H_i$ does not intersect with the augmenting path $P$, 
then the augmentation does not affect 
$G^\circ[H_i]$, and the algorithm simply keeps the same routing and ordering as before. 

For each blossom $H_i \in \Lambda_{\rm n}$ with $H_i \cap P \not= \emptyset$, in the order of indices, we apply $\Search$ to $H_i \cup \{b_i\}$ 
in which we regard  the dummy line $\{b_i,t_i\}$ as the unique source line.
The family of blossoms is restricted to the set of blossoms $H_j\in\Lambda_{\rm n}$ with $H_j\subsetneq H_i$. 
For each inner blossom $H_j$, we have already computed $<_{H_j}$ and $R_{H_j}(x)$ for $x \in H^\bullet_j$. 
Since there exists no augmenting path in $H_i$, $\Search$ always returns $\emptyset$. 
Then, we can show that 
the procedure labels every vertex in $H_i \cap V$ without updating the dual variables as we will see in Lemma~\ref{lem:reachableHi}. 
However, this procedure may create new blossoms in $H_i$, 
and the bud $b$ of such a blossom $H$ is not labeled.  
This means that we do not obtain $R_{H_i}(b)$, whereas $b$ might be in $H^-_i$. 
To overcome this problem, we update the dual variables and apply $\Expand(H_i)$. 
Whenever $\Search$ terminates, we update the dual variables as we will describe later.  
We repeat this process until $q (H_i)$ becomes zero. 
Then, we apply $\Expand(H_i)$. 
 
A new blossom $H$ created in this procedure
is accompanied by $<_{H}$ and $R_{H}(x)$ for $x \in H^\bullet$ satisfying (BR1)--(BR5)
by the argument in Sections~\ref{sec:search} and \ref{sec:validity}. 
We can also see that 
$p$ and $q$ are feasible after creating a new blossom 
by the same argument as Lemma~\ref{lem:createblossomdual}.

This argument shows that (BT1), (BT2), and (DF1)--(DF3) hold when we restrict the instance to $H_i \cup \{b_i\}$.
We now show that we can create a new blossom $H$ with $q(H) = 0$ in the procedure
so that these conditions hold in the entire instance. 
To this end, when we create a new blossom $H$, 
we define the row and the column of $C^*$ corresponding to $\{b, t\}$ as follows. 
\begin{itemize}
\item
If $b \in B^*$, then 
we define $C^*_{b y}= 0$ for any $y \in (V^* \setminus (H_i \cup \{b_i\})) \setminus B^*$ and
$C^*_{x t} = C^*_{x g}$ for any $x \in B^* \setminus (H_i \cup \{b_i\})$ (see Fig.~\ref{fig:revisionfig30}).
\item
If $b \in V^* \setminus B^*$, then 
we define $C^*_{x b}= 0$ for any $x \in B^* \setminus (H_i \cup \{b_i\})$ and
$C^*_{t y} = C^*_{g y}$ for any $y \in (V^* \setminus (H_i \cup \{b_i\})) \setminus B^*$
\item
The other entries in $C^*$ are determined by $\Search$ in $H_i \cup \{b_i\}$.
\item
Then, apply the pivoting operation to $C^*$ around $\{b, t\}$. 
\end{itemize}
In other words, we consider all the vertices in $V^*$ (instead of $H_i \cup \{b_i\}$) 
when we introduce new vertices in Step 2 of $\Blossom(v, u)$ or Step 1 of $\DBlossom(v, H_i)$. 
Note that this modification does not affect the entries in $C^*[H_i \cup \{b_i\}]$, 
and hence it does not affect $\Search$ in $H_i \cup \{b_i\}$.

\begin{figure}[htbp]
  \centering
    \includegraphics[width=12cm]{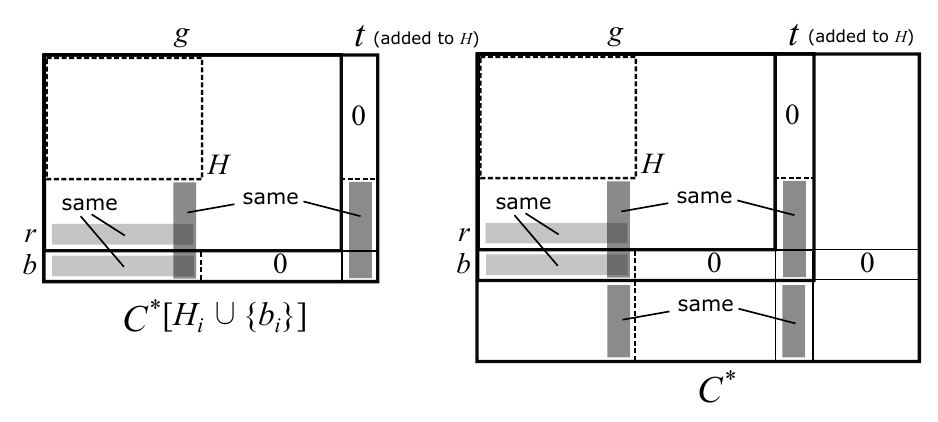}
      \caption{Definition of $C^*$. Each element in the left figure is determined by $\Search$ in $H_i \cup \{b_i\}$
                   in the same way as Fig~\ref{fig:revisionfig13}. We extend this definition to the entire matrix as shown in the right figure.} 
    \label{fig:revisionfig30} 
\end{figure}

\begin{lemma}
\label{lem:reachableHi}
When we apply $\Search$ in $H_i \cup \{b_i\}$ as above,
the procedure labels every vertex in $H_i \cap V$ without updating the dual variables. 
\end{lemma}

\begin{proof}
By Lemma~\ref{lem:labeliff}, it suffices to show that 
for every vertex 
$v \in H_i \cap V$,
there exists a vertex set $X \subseteq H_i$ with the conditions in Lemma~\ref{lem:labeliff}. 

We first show that such a vertex set exists after Step 3 of $\Augment(P)$. 
For a given vertex $v\in H_i \cap V$, define $Z \subseteq H_i$ by 
$$Z := 
\begin{cases}
R_{H_i}(v) \setminus \{v\} & \mbox{if $v \not\in P$}, \\
R_{H_i}(\bar v) \setminus \{\bar v \} & \mbox{if $v \in P$}.
\end{cases}
$$ 
Then, $G^\circ[Z]$ has a unique tight perfect matching by (BR4). 
Set
$$
Y:= Z \cup (P^* \setminus H_i) \cup \{b_j, t_j \mid H_j \in \Lambda_{\rm n},\ H_j \subsetneq H_i,\ H_j \cap Z = \emptyset\}. 
$$
Since each of $G^\circ[P^* \setminus H_i]$ and $G^\circ[P^* \cap H_i]$ has 
a unique tight perfect matching, where $P^* \cap H_i$ might be the emptyset, 
$G^\circ[Y]$ also has a unique tight perfect matching. 
This shows that $C^*[Y]$ is nonsingular before the pivoting operation around $P^*$ in Step 2 of $\Augment(P)$ by Lemma~\ref{lem:tightmatching}, and hence 
$C^*[X]$ with $X := Y \triangle P^*$ is nonsingular after the pivoting operation around $P^*$ by Lemma~\ref{lem:pivotsing}. 
Then, $X \cup \{v\}$ consists of lines, dummy lines, and the tip $t_i$, and 
it contains all the buds and the tips in $H_i$ after updating $b_i$ and $t_i$ in Step 3 of $\Augment(P)$.  
Furthermore, the tightness of the perfect matching in $G^\circ[Y]$ shows that 
$X$ satisfies (\ref{eq:8411}) after the augmentation. 
Thus, $X$ satisfies the conditions in Lemma~\ref{lem:labeliff} after Step 3 of $\Augment(P)$.

We next show that such a set $X$ exists after Steps 4 and 5 of $\Augment(P)$. 
Suppose that we apply (i)--(iv) in Step 4 or 5 of $\Augment(P)$ for $H_j \in \Lambda_{\rm n}$, 
that is, we apply the pivoting operation around  $X_j := \{b_j, t_j, b'_j, t'_j\}$. 
We consider the following three cases, separately. 
\begin{itemize}
\item
Suppose that $H_j \subsetneq H_i$. 
In this case, since $C^*[X]$ is nonsingular before the pivoting operation around $X_j$, 
$C^*[X \triangle X_j]$ is nonsingular after the pivoting operation by Lemma~\ref{lem:pivotsing}. 
We can also check that $X \triangle X_j$ satisfies the other conditions in Lemma~\ref{lem:labeliff}. 

\item
Suppose that $H_j \supsetneq H_i$ or $H_j \cap H_i = \emptyset$. 
Let $X':= X \cup \{b_j, t_j\}$. 
Since $X'$ satisfies (\ref{eq:8411}) and $|X' \cap H_i|$ is even, 
the nonsingularity of $C^*[X]$ and $C^*[\{b_j, t_j\}]$ shows that 
$C^*[X']$ is nonsingular before the pivoting operation around $X_j$. 
Hence, $C^*[X' \triangle X_j]$ is nonsingular after the pivoting operation by Lemma~\ref{lem:pivotsing}. 
Since $|X' \cap H_i|$ is even, this implies that 
$C^*[(X' \triangle X_j) \setminus \{b'_j, t'_j\}] = C^*[X]$ is nonsingular after the pivoting operation.
We can also check that $X$ satisfies the other conditions in Lemma~\ref{lem:labeliff}. 

\item
Suppose that $H_j = H_i$. Let $X':= X \cup \{b_i, b'_i\}$. 
Since $(b_i, b'_i) \in F^*$ and there is no edge in $F^*$ between $b'_i$ and $H_i$, 
the nonsingularity of $C^*[X]$ shows that $C^*[X']$ is nonsingular
before the pivoting operation around $X_j$. 
Hence, $C^*[X' \triangle X_j]$ is nonsingular after the pivoting operation by Lemma~\ref{lem:pivotsing}. 
We can also check that $X' \triangle X_j$ satisfies the other conditions in Lemma~\ref{lem:labeliff}. 
\end{itemize}
By these cases, 
there exists a set $X$ satisfying the conditions in Lemma~\ref{lem:labeliff} after Steps 4 and 5 of $\Augment(P)$. 

Therefore, every vertex in $H_i \cap V$ is labeled without updating the dual variables by Lemma~\ref{lem:labeliff}, 
which completes the proof. 
\end{proof}

In what follows in this subsection, 
we describe how to update the dual variables. 

Suppose that $\Search$ returns $\emptyset$ when it is applied to $H_i \cup \{b_i\}$. 
Define $R^+$, $R^-$, $Z^+$, $Z^-$, $Y$, and $\epsilon=\min\{\epsilon_1,\epsilon_2,\epsilon_3,\epsilon_4\}$ 
as in Section~\ref{sec:dualupdatealgo}. By Lemma~\ref{lem:reachableHi}, 
we have that $R^+ = \{b_i, t_i \} $, $R^-=\{b_j\mid \mbox{$H_j$: maximal blossom with $H_j\subsetneq H_i$}\}$, 
$Z^-= Y = \emptyset$, and $\epsilon_2 = \epsilon_3 = \epsilon_4 = + \infty$. 
In particular, every maximal blossom is labeled with $\oplus$.
Here, a blossom $H_j\subsetneq H_i$ is called a {\em maximal blossom} if there exists no blossom $H$ with $H_j \subsetneq H \subsetneq H_i$.
We now modify the dual variables in $V^*$ as follows. 
Set $\epsilon' := \min \{ \epsilon,  q (H_i)\}$, which is a finite positive value. 
Then update $p(t_i)$ as 
$$p(t_i) :=\begin{cases}
 p(t_i)+\epsilon' & (t_i \in  B^*), \\ 
 p(t_i)-\epsilon' & (t_i \in  V^* \setminus B^*), 
\end{cases}$$
and update $q(H_i)$ as $q(H_i) := q(H_i)-\epsilon'$. 
For each maximal blossom $H_j\subsetneq H_i$,  
which must be labeled with $\oplus$, 
update $q(H_j)$ as $q(H_j) := q(H_j)+\epsilon'$ 
and $p(b_j)$ as 
$$
p(b_j) :=\begin{cases}
 p(b_j) - \epsilon' & (b_j \in  B^*), \\ 
 p(b_j) + \epsilon' & (b_j \in  V^* \setminus B^*). 
\end{cases}
$$
Note that $\Expand(H_j)$ is not applied for any maximal blossom $H_j \subsetneq H_i$, 
because $q(H_j) > 0$ after the dual update, 
whereas $\Expand(H_i)$ is applied when $q (H_i)$ becomes zero.

We now prove the following claim, which shows the validity of this procedure. 

\begin{claim}
The obtained dual variables $p$ and $q$ are feasible in $V^*$  (not only in $H_i$).  
\end{claim}

\begin{proof}
It suffices to show (DF2). Suppose that $u \in B^*$, $v \in V^* \setminus B^*$, and $(u, v) \in F^*$. 
Since the value of $q(H)$ is zero for newly created blossoms $H$, the dual variable $(p,q)$ are feasible 
at the end of in $\Search$ applied to $H_i \cup \{b_i\}$. Updating the dual variables decreases the slack 
$p(v)-p(u)-Q_{uv}$ only if $u$ and $v$ belong to distinct maximal blossoms included in $H_i$ or 
one of them is $t_i$. In these cases, however, we have $p(v)-p(u)-Q_{uv}\geq\epsilon_1\geq\epsilon'$. 
Thus the above update of the dual variables does not violate the feasibility.   
\end{proof}

\section{Algorithm Description and Complexity}
\label{sec:complexity}

Our algorithm for the minimum-weight parity base problem is described as follows. 

\begin{description}
\item[Algorithm] {$\sf Minimum$-$\sf Weight$ $\sf Parity$ $\sf Base$} 
\item[Step 1:]
Split the weight $w_\ell$ into $p(v)$ and $p(\bar{v})$ for each 
line $\ell=\{v,\bar{v}\}\in L$, i.e., $p(v) + p(\bar{v}) = w_\ell$. 
Execute the greedy algorithm for finding a base $B\in\B$ 
with minimum value of $p(B)=\sum_{u\in B}p(u)$. 
Set $\Lambda = \emptyset$. 

\item[Step 2:]
If there is no source line, then return $B := B^* \cap V$ as an optimal solution. 
Otherwise, apply $\Search$. 
If $\Search$ returns $\emptyset$, then go to Step 3. 
If $\Search$ finds an augmenting path $P$, then go to Step 4. 

\item[Step 3:]
Update the dual variables as in Section~\ref{sec:dualupdatealgo}. 
If $\epsilon = +\infty$, then conclude that there exists no parity base and terminate the algorithm. 
Otherwise, apply $\Expand(H_i)$ for all maximal blossoms $H_i$ with $q(H_i)=0$ and go to Step 2.

\item[Step 4:]
Apply $\Augment(P)$ 
to obtain a new base $B^*$, a family $\Lambda$ of blossoms, 
and feasible dual variables $p$ and $q$. 
For each normal blossom $H_i$ with $H_i \cap P \not= \emptyset$ in the increasing order of $i$, do the following. 
\begin{quote}
While $q (H_i) > 0$, apply $\Search$ in $H_i$ and update the dual variables as in Section~\ref{sec:reroute}. 
Apply $\Expand(H_i)$. 
\end{quote}
Go back to Step 2. 
\end{description}

We have already seen the correctness of this algorithm, and 
we now analyze the complexity. 
Since $|V^*| \:= O(n)$, an execution of 
the procedure $\Search$ as well as the dual update requires $O(n^2)$ arithmetic operations.  
By Lemma~\ref{lem:dualupdatebound}, 
Step 3 is executed at most $O(n)$ times per augmentation. 
In Step 4, we create a new blossom or apply $\Expand(H_i)$ when we update the dual variables, 
which shows that the number of dual updates as well as executions of $\Search$ in Step 4 is also bounded by $O(n)$. 
Thus, $\Search$ and dual update are executed $O(n)$ times per augmentation, which requires $O(n^3)$ operations. 
We note that it also requires $O(n^3)$ operations to update $C^*$ and $G^*$ after augmentation. 
Since each augmentation reduces the 
number of source lines by two, the number of augmentations during the algorithm is $O(m)$, where $m=\rank A$, 
and hence the total number of arithmetic operations is $O(n^3 m)$.

\begin{theorem}
\label{th:complexity}
Algorithm {\sf Minimum-Weight Parity Base} finds a parity base of minimum weight or 
detects infeasibility with $O(n^3m)$ arithmetic operations over $\K$. 
\end{theorem}

If $\K$ is a finite field of fixed order, each arithmetic operation can be executed in $O(1)$ time. 
Hence Theorem~\ref{th:complexity} implies the following. 

\begin{corollary}
The minimum-weight parity base problem over 
an arbitrary fixed finite field $\K$ can be solved 
in strongly polynomial time.
\end{corollary}

When $\K = \mathbb{Q}$, it is not obvious that 
a direct application of our algorithm runs in polynomial time. 
This is because we do not know how to bound the number of bits required to represent the entries of $C^*$. 
However, the minimum-weight parity base problem over $\mathbb Q$ can be solved 
in polynomial time by applying our algorithm over a sequence of finite fields.

\begin{theorem}
The minimum-weight parity base problem over $\mathbb Q$ can be solved 
in time polynomial in the binary encoding length $\langle A \rangle$ of the matrix representation $A$.
\end{theorem}

\begin{proof}
By multiplying each entry of $A$ by the product of the denominators of all entries, 
we may assume that each entry of $A$ is an integer. 
Let $\gamma$ be the maximum absolute value of the entries of $A$, 
and put $N := \lceil m \log (m \gamma) \rceil$. 
Note that $N$ is bounded by a polynomial in $\langle A \rangle$. 
We compute the $N$ smallest prime numbers $p_1, \dots , p_N$. 
Since it is known that $p_N = O(N \log N)$ by the prime number theorem, 
they can be computed in polynomial time by the sieve of Eratosthenes.

For $i=1, \dots , N$, we consider the minimum-weight parity base problem over ${\rm GF}(p_i)$
where each entry of $A$ is regarded as an element of ${\rm GF}(p_i)$. 
In other words, we consider the problem in which each operation is executed modulo $p_i$. 
Since each arithmetic operation over ${\rm GF}(p_i)$ can be executed in polynomial time, 
we can solve the minimum-weight parity base problem over ${\rm GF}(p_i)$
in polynomial time by Theorem~\ref{th:complexity}. 
Among all optimal solutions of these problems, the algorithm returns the best one $B$. 
That is, $B$ is the minimum weight parity set 
subject to $|B| = m$ and $\det A[U, B] \not\equiv 0 \pmod {p_i}$ for some $i \in \{1, \dots , N\}$. 

To see the correctness of this algorithm, 
we evaluate the absolute value of the subdeterminant of $A$. 
For any subset $X \subseteq V$ with $|X| = m$, we have 
$$
|\det A[U, X]| \le m! \gamma^m \le (m \gamma)^m \le 2^N < \prod^N_{i=1} p_i. 
$$
This shows that 
$\det A[U, X] = 0$ if and only if $\det A[U, X] \equiv 0 \pmod {\prod^N_{i=1} p_i}$. 
Therefore, 
$\det A[U, X] \not= 0$ if and only if $\det A[U, X] \not\equiv 0 \pmod {p_i}$ for some $i \in \{1, \dots , N\}$, 
which shows 
that the output $B$ is an optimal solution. 
\end{proof}

\section*{Acknowledgements}
The authors thank the anonymous reviewer and Kei Nakashima for very careful reading of our manuscript and valuable comments.
They also thank Jim Geelen, Gyula Pap and Kenjiro Takazawa for 
fruitful discussions on the topic of this paper. 
This work is supported by JST through CREST, No.~JPMJCR14D2, 
ACT-I, No.~JPMJPR17UB, and ERATO, No.~JPMJER1201, and 
by Grants-in-Aid for Scientific Research 
No.~JP24106002 and No.~JP24106005 from MEXT, and No.~JP16K16010 from JSPS.

\end{document}